\documentclass[aps,10pt,notitlepage]{revtex4-1}
\usepackage[margin=1in]{geometry}

\usepackage{mystyle}
\usepackage{array}

\begin{document}
\title{Fault-tolerant gates on hypergraph product codes}
\author{Anirudh Krishna}\
\affiliation{
 D\'epartement de physique \& Institut Quantique, Universit\'e de Sherbrooke, Sherbrooke, Qu\'ebec J1K 2R1, Canada
}
\author{David Poulin}
\affiliation{
 D\'epartement de physique \& Institut Quantique, Universit\'e de Sherbrooke, Sherbrooke, Qu\'ebec J1K 2R1, Canada
}

\date{\today}
\begin{abstract}
Hypergraph product codes are a class of quantum low density parity check (LDPC) codes discovered by Tillich and Z\'emor.
These codes have a constant encoding rate and were recently shown to  have a constant fault-tolerant error threshold.
With these features, they asymptotically offer a smaller overhead compared to topological codes.
However, existing techniques to perform logical gates in hypergraph product codes require many separate code blocks, so only becomes effective with a very large number of logical qubits.
Here, we demonstrate how to perform Clifford gates on this class of codes using code deformation.
To this end, we generalize punctures and wormhole defects, the latter introduced in a companion paper.
Together with state injection, we can perform a universal set of gates within a single block of the class of hypergraph product codes.
\end{abstract}
\maketitle

\section{Introduction}
Quantum error correcting codes can be used to simulate an ideal quantum circuit using noisy circuit components \cite{aharonov1997fault,aliferis2006quantum,kitaev1997quantum,knill1998resilient}.
These simulations are characterized by a tug-of-war between the size of the circuit and the accuracy of the result.
Larger problem instances require larger circuits, and therefore larger error correcting codes; this however creates more opportunities for errors to accumulate.
To compensate for this increased error accumulation, we must lower the logical fault rate.
By choosing our quantum error correcting code appropriately, we can increase the code size $n$ logarithmically with the size of the logical circuit and still find that the global logical error rate decreases exponentially in the code size $n$.

Not all error correcting codes are equivalent.
The size of the noisy circuit depends on the quantum error correcting code being used.
Finding codes which minimize the size of the circuit to achieve a target logical error rate is a subject of active research.

A large amount of work in both theory and experiment focuses on topological codes \cite{kitaev2003fault,bravyi1998quantum,bombin2006topological}.
These codes have the desired ability to suppress errors exponentially in the size $n$ of the codes.
Furthermore, topology guarantees that each qubit, be it data qubit or ancilla qubit for readout, is only connected to a constant number of other qubits in its neighborhood.
These properties, among others, make these codes suitable for serving as the architecture for quantum computers.

One drawback of topological codes is their ability to store logical qubits.
Indeed, each logical qubit in a topological code is encoded in a distinct block of size $n$. 
The number of physical qubits required to simulate a single logical qubit therefore increases with the size of the quantum computer.

In contrast, quantum \emph{low density parity check} (LDPC) codes \cite{kovalev2013fault,gottesman2014fault} are generalizations of topological codes that overcome this increasing encoding overhead.
LDPC codes refer to families of codes where all qubits (data qubits and ancilla qubits for readout) are only connected to a constant number of other qubits.
This constant is independent of the block size and thus simplifies the process of syndrome extraction.
In \cite{gottesman2014fault}, Gottesman proposed a construction that combines techniques for efficient syndrome extraction with ideas to perform logical gates on block codes.
The result was a conditional statement: if `good' quantum LDPC codes exist, the number of physical qubits required to simulate a logical qubit becomes a constant, independent of the size of the quantum computer.

The difference between LDPC codes and topological codes is that the connectivity need no longer be spatially local.
By sacrificing locality, these codes overcome one of the shortcomings of the surface code and this permits a `wholesale effect'.
Increasing the size $n$ of the code lets us encode $k$ qubits, where $k$ can increase with $n$.
We can find codes for which the logical error probability decreases exponentially with $n$ with a fixed $k/n$ ratio. 
This is to be contrasted with topological codes that also achieve an exponential error suppression with $n$, but with $k=1$ and hence vanishing encoding rate $1/n$. 

Good LDPC codes are elusive.
It is hard to enforce the commutation relations between stabilizers of a quantum code while simultaneously maintaining low connectivity.
Topological codes use topology to achieve this, but there are strong constraints on the number of logical qubits they can simulate and how effectively they can do so \cite{bravyi2009no,bravyi2010tradeoffs}.
These constraints are a consequence of \emph{locality}.
Although no doubt simpler to engineer, locality is not a fundamental constraint.
There exist techniques that permit qubits that are not adjacent to share entanglement in various architectures \cite{nickerson2013topological,campagne2018deterministic,axline2018demand,kurpiers2018deterministic}.
This motivates theoretical investigations of LDPC codes that are not constrained by locality.
Only with a complete understanding of the potential benefits of LDPC codes will we be able to decide if they are worth the extra experimental effort. 

One of the leading candidates for LDPC codes are the so-called hypergraph product codes which eschew topology, and instead engineer commutation relations using algebraic / graph-theoretic techniques \cite{tillich2014quantum}.
The hypergraph product is itself not a code family, but rather a technique to construct quantum codes from classical codes.
If we input two classical codes to the hypergraph product machinery, the resulting quantum code inherits properties of the classical codes.
Importantly, if the classical code families are LDPC then the quantum code families will also be LDPC.
Furthermore, if the classical code families have a code dimension $k$ scaling linearly in the block size $n$, then so does the code dimension of the quantum code family.
With regards to distance, the hypergraph product code construction yields quantum codes with distance scaling as the square-root of the block size.
Up to constants, this is the same functional dependence between the distance and the block size as the surface code (but it applies to all the logical qubits).
So in short, quantum LDPC code achieve the same error suppression as topological codes, but do so at a constant encoding rate.

When these codes were first proposed, it was unclear if they possessed an efficient decoding algorithm.
Naively applying decoders for classical codes to their quantum counterparts does not work because of degeneracy \cite{poulin2008iterative};
there exist low-weight errors that confuse the classical decoding algorithm.
We needed an inherently quantum decoding algorithm, but it was not clear if one existed.
Since then, the situation has changed significantly.

Leverrier et al. \cite{leverrier2015quantum} were the first to show that we could overcome this problem.
For certain classes of hypergraph product codes called quantum expander codes, they discovered an algorithm called small-set-flip that overcame the degeneracy issue.
The small-set-flip algorithm was shown to function in linear time and is therefore efficient.
Secondly, it was shown that the small-set-flip algorithm came with a performance guarantee.
Suppose an adversary were handed the ability to target qubits and apply a Pauli error of choice on each qubit to inflict maximum damage.
The small-set-flip algorithm can protect against errors whose weight is less than the square-root of the block size (up to some constant).

Of course, this is a pessimistic error model; errors that occur in nature may not be orchestrated by an adversary.
The local stochastic error model is one way to model errors that occur naturally.
Roughly, it states that the probability of an error decays exponentially with the number of qubits afflicted.
Fawzi et al. showed that the small-set-flip algorithm is capable of correcting a constant fraction of such errors \cite{fawzi2018efficient}.

In a followup work, Fawzi et al. demonstrated that these codes are even resilient to syndrome errors \cite{fawzi2018constant}; we only need to measure the syndromes once in order to proceed with decoding.
If we applied the small-set-flip algorithm to decode in the presence of potentially incorrect syndromes, then it is guaranteed that the number of physical errors on the qubits is upperbounded by a function of the number of syndrome errors.
Thus so long as we do not have too many syndrome errors, we are guaranteed that we can reduce the weight of the error on the physical qubits.

These works put the hypergraph product code on solid theoretical foundation.
In trying to reason about its threshold however, the best bounds were very weak.
Numerical work has sought to remedy this.
Kovalev et al. \cite{kovalev2018numerical} estimated the threshold of the hypergraph product codes using a statistical mechanical mapping.
These computations are independent of a decoder, and instead estimate the threshold using a proxy.
Grospellier and Krishna \cite{grospellier2018numerical} estimated the threshold of the hypergraph product code using the small-set-flip algorithm.
The focus of this work was a class of hypergraph product codes with qubit degree $11$ and check degree $10$ and $12$.
These codes were subject to an independent bit and phase flip error model.
The estimates for the thresholds are several orders of magnitudes better than previous analytical lower bounds.

On a related note, Liu and Poulin have studied small codes using a neural-network \cite{liu2018neural}.
Finally, Panteleev and Kalachev \cite{panteleev2019degenerate} showed very promising results for related codes.
The observed threshold can be comparable to the one observed for the surface code, but in general will depend strongly on the code parameters such as the encoding rate and the weight of the checks.

These results demonstrate that hypergraph product codes are efficient fault-tolerant quantum memories.
In this paper, we describe the first techniques to perform Clifford gates fault tolerantly on hypergraph product codes.
Our method is an instance of code deformation, a general framework to perform gates on quantum codes.
Continuing in the spirit of the hypergraph product construction, we express defects on the surface code as purely algebraic and graph-theoretic concepts.
Importantly, the code remains LDPC over the course of code deformation.
Fault tolerance follows because the modifications we make at each step are local (in a graph-theoretic sense).
The generalized defects are capable of encoding several logical qubits.
If we choose to use a subset of these qubits to encode information, then the rest can be considered as gauge qubits.
We discuss constraints on code deformation that keeps the spaces of logical and gauge qubits separate.
To conclude, we show that we can achieve a universal gate set on the logical qubits via state injection.

We emphasize at this point that we are not providing a technique to \emph{compile} a Clifford gate of interest.
We provide a framework within which it is possible to realize Clifford gates via code deformation.
These Clifford gates can then be composed to generate a larger group of transformations.
We show that the framework is sufficiently rich to realize all different types of generating gates, but depending on the code, 
 this may or may not encompass the entire space of Clifford operators.
We will return to this discussion later.

This approach can be contrasted to Gottesman's work which entails partitioning logical qubits into blocks of LDPC codes.
Each block is of `intermediate' size and a computation with $k$ logical qubits requires blocks whose size scales as $O(k/\operatorname{poly}\log(k))$.
Gates are then performed by state injection using ancilla states.
The size of these blocks limits the number of gates that can be implemented at any given time.
Furthermore, the savings of LDPC codes become compelling only as we increase the block size.
Partitioning the logical qubits into blocks implies that it will take longer for this effect to manifest.
In contrast, we propose performing quantum computation on a single block.
Our proposal does not limit the number of qubits that can be processed at any given time to a constant.
On the other hand, the time required to perform a gate could scale so it is unclear whether these gates will be faster.

\textbf{Outline of the paper:}
In section \ref{sec:background}, we begin by reviewing some facts about classical codes and proceed to recall the definition of the hypergraph product code.

In section \ref{sec:punctures}, we describe how to deform the hypergraph product code by introducing a {\em puncture}.
This puncture shall itself be described as a hypergraph product of subgraphs of the graphs that together form the quantum code.
Section \ref{subsec:defPunc} includes the definition of a puncture, and describes how they arise in two types, smooth and rough.
In section \ref{subsec:logicalPaulis}, we describe the logical operators supported on the puncture.

We generalize puncture defects in section \ref{sec:wormholes}, and discuss how to create a wormhole.

Unlike a surface code, a puncture will be capable of supporting several logical qubits.
In section \ref{sec:code-def}, we study how to perform code deformation.
We first impose constraints on code deformation with multiple logical qubits in order to guarantee that the process is not error prone in section \ref{subsec:nonmixing}.
We proceed to list the requirements to perform all Clifford gates on the hypergraph product code in section \ref{subsec:codedefhgp}.
After discussing state injection in section \ref{subsec:injection}, we conclude by discussing point-like punctures in \ref{subsec:pointlike}.
Point-like punctures are useful in demonstrating how the notion of topology may generalize to purely graph-theoretic conditions.
From a condensed-matter physics perspective, these punctures can be seen as a generalization of anyons from two-dimensionnal manifolds to general graph structures, which could be of independent interest.

\section{Background and notation}
\label{sec:background}

\subsection{Classical and quantum codes}

\textbf{Classical codes:} A classical code $\C = [n,k,d] \subseteq \field_2^n$ over $n$ bits is the (right-)kernel of a matrix $\h \in \field_2^{m \times n}$, known as its parity check matrix.
The code dimension $k$ is the dimension of the kernel, $\dim{\ker{\h}}$, and $d$ is the minimum Hamming distance between a pair of vectors in $\C$.
In general, we could have redundant checks and so $m \geq n - k$.
We let $\rs{\h}$ denote the rowspan of the matrix $\h$ and $\h^t$ be the transpose of $\h$.
Each row of $\h$ encodes a parity constraint that we refer to as a check and we label $c\in C$.
Likewise we label the columns of $\h$ by variable indices $v\in V$.
Let $\1_V, \1_C$ denote the identity on $\field_2^V$ and $\field_2^C$ respectively.
For $u,v \in \field_2^n$, we let $\ip{u}{v} = \sum_{i=1}^{n} u_i v_i$ denote the inner product between them.

An LDPC code is a code family $\{\C_n\}_n$ such that as a function of the block size $n$, the number of bits in the support of a check is upper bounded by a constant, as are the number of checks a bit is connected to \cite{richardson2008modern}.
In other words, the number of non-zero elements in each row and column of the parity-check matrix is bounded by a constant with respect to the block-size $n$.
The factor graph $\G(\C)$ of a code $\C$ lets us infer properties of the code $\C$ from the properties of the graph.
The graph $\G(\C) = (V \union C, E)$ is a bipartite graph, where $V = [n]$ and $C = [m]$.
We draw an edge between check node $c \in C$ and variable node $v \in V$ if and only if $\h_{cv} = 1$.
Given a subset $P \subseteq V\union C$, the neighborhood of $P$ is denoted $\Gamma(P)$ and is defined as
\begin{align*}
  \Gamma(P) = \{q | (q,p) \text{ or } (p,q) \in E \text{ for } p \in P \}~.
\end{align*}

\textbf{Quantum codes:} Let $\P = \{I,X,Y,Z\}$ denote the Pauli group and $\P_n = \P^{\otimes n}$ denote the $n$-fold tensor product of the Pauli group.
A quantum error correcting code is specified by a group $\Q \subseteq \P_n$.
The stabilizer $\S$ of the code is the center $Z(\Q)$, and the quotient group $\G := \Q/\S$ is called the (pure) gauge group.
The set of logical operators is then defined as $\N(\Q) \setminus \Q$.
A stabilizer code is a code such that $\G = \emptyset$ and $\Q$ is an Abelian group.

Just like the individual rows of the classical parity check matrix generate a linear space of constraints, we can choose a generating set of checks for the stabilizer group $\S$.
Much like its classical counterpart, the factor graph $\G_Q$ can be use to represent the stabilizer generator, and provides a visual representation of $\Q$.
The only difference is that the edges could carry labels of Pauli elements $X, Y$ or $Z$, indicating the action of a check on a qubit.

\subsection{The hypergraph product code}
The hypergraph product code is a way to construct a quantum code $\Q$ given two classical codes $\C_1$ and $\C_2$.
This can naturally be extended to families of classical codes.
If the two classical code families are LDPC, then so is the resulting quantum code family.
The resulting code is a CSS code \cite{calderbank1996good,steane1996multiple}, i.e. the stabilizer generators are either  products of only $X$ operators or only $Z$ operators.
For simplicity, we shall consider the graph product of a code $\C$ with itself.

\textbf{Graph-theoretic description:} Let $\G = (V \union C, E)$ be a bipartite graph.
Let $\Q$ denote the quantum code obtained from the hypergraph product of $\G$ with itself.
The factor graph $\G_\Q$ of $\Q$ is defined as $\G_\Q = \G \times \G$.
Its nodes are partitioned as follows:
\begin{enumerate}
  \item qubits $V \times V \union C \times C$;
  \item $X$ stabilizers $V \times C$;
  \item $Z$ stabilizers $C \times V$.
\end{enumerate}
This representation highlights that this construction yields two kinds of qubits -- those emerging from the product of two variable nodes (VV nodes) and those emerging from the product of two check nodes (CC nodes).
We draw an edge between $(a_1,a_2)$ and $(b_1,b_2)$ in $V \union C \times V \union C$ if either $(a_1,b_1) \in E$ and $a_2 = b_2$ or if $(a_2,b_2) \in E$ and $a_1 = b_1$.

\textbf{Algebraic description:} Let $\h \in \field_2^{m \times n}$ define the codes $\C = [n,k,d]$ and $\tC = [m, \tk, \td]$ as
\begin{align}
  \begin{matrix}
  \C = \ker{\h} & \tC = \ker{\h^t}~.
  \end{matrix}
\end{align}
The $X$ and $Z$ stabilizers of the code are specified via their symplectic representation \cite{nielsen2002quantum}.
The parity check matrices of the quantum code are denoted $\h_X$ and $\h_Z$ respectively, where
\begin{align}
  \label{eq:graphProductDef}
  \h_X = \left(\1_{V} \otimes \h| \h^t \otimes \1_{C} \right) \qquad
  \h_Z = \left(\h \otimes \1_{V}| \1_{C} \otimes \h^t \right) ~.
\end{align}
In this expression, VV nodes are in the left partition while CC nodes are in the right partition.
Let $d_{\min} = \min\{d,\td\}$ denote the minimum of the distance of the two codes.
The hypergraph product $\Q$ is a $\dsl n^2 + m^2, k^2 + \tk^2, d_{\min} \dsr$ quantum code.

We shall refer to the logical operators of the quantum code $\Q$ as the embedded logical operators to distinguish them from the logical operators that we introduce later by creating defects.
The embedded logical operators of the code are described as follows.
\begin{lemma}\textbf{(Embedded logical operators)}
  \label{lem:embeddedLogicals}
    \begin{enumerate}
      \item the $X$ logical operators of $\Q$ are spanned by
         \[
          \left(\ker{\h} \otimes \left(\field_2^n / \rs{\h}\right) |\bz_{m^2}\right) \union
          \left(\bz_{n^2} \;|\left(\field_2^m / \rs{\h^t}\right) \otimes \ker{\h^t}\right)
         \]
      \item the $Z$ logical operators of $\Q$ are spanned by
        \[
          \left( \left(\field_2^n / \rs{\h}\right) \otimes \ker{\h}|\bz_{m^2}\right) \union
          \left(\bz_{n^2} \;| \ker{\h^t} \otimes \left(\field_2^m / \rs{\h^t}\right) \right)
        \]
    \end{enumerate}
\end{lemma}
\begin{proof}
  The style of the proof follows arguments presented in lemma 17 of \cite{tillich2014quantum}.
  We first show that the spaces above are contained in the set of logical operators, and then use counting arguments to show that this must be the entire space of logical operators.

  We deal with the $X$ type logical operators and note that the $Z$ logical operators follow using a similar argument.
  Let $\alpha$ be an $X$ logical operator, i.e.
  \begin{align*}
    \alpha \in \left(\ker{\h} \otimes \left(\field_2^n / \rs{\h}\right) |\bz_{m^2}\right) \union
          \left(\bz_{n^2} \;|\left(\field_2^m / \rs{\h^t}\right) \otimes \ker{\h^t}\right)~.
  \end{align*}
  This object clearly commutes with the $Z$ stabilizers.

  For the sake of contradiction, assume that $\alpha$ is in fact in the span of the $X$ stabilizers, i.e. that there exists a non-trivial vector $a \in \field_2^{V \times C}$ such that $a\h_X = \alpha$.
  Without loss of generality, let us assume that the VV portion of $\alpha$ is non-trivial and let $\pi(\alpha)$ be the projection of $\alpha$ on to the VV type qubits.

  It follows that
  \begin{align}
  \label{eq:Xlogicalcontra2}
    a(\1_V \otimes \h) = \pi(\alpha)~.
  \end{align}
  For $u,v \in V$, we can index the elements of $\pi(\alpha)$ as $\pi(\alpha)[u,v]$.
  Furthermore, for fixed $u \in V$, we let $\pi(\alpha)[u,*]$ denote the vector over $\field_2^V$ obtained by fixing the first component of $\pi(\alpha)$.

  Similarly, we can index the elements of $a$ as $a[v,c]$ for $v \in V$ and $c\in C$ and let $a[v,*]$ denote the vector over $\field_2^C$.
  Eq. \ref{eq:Xlogicalcontra2} implies that there exists some index $u \in V$ such that
  \begin{align*}
    b_u \h = \beta_u~,
  \end{align*}
  where $b_u := a[u,*]$ and $\beta_u := \pi(\alpha)[u,*]$.
  However, this is a contradiction since $\beta_u \in \field_2^V/\rs{\h}$ and lies outside the row-span of $\h$.
  
  The row rank of $\h$ is $n-k$ and so the number of cosets in $\field_2^n/\rs{\h}$ is $n - (n-k) = k$.
  Therefore the number of elements in $\ker{\h} \otimes \field_2^n/\rs{\h}$ is $k^2$.
  Similarly, the number of cosets $\field_2^m/\rs{\h^t}$ is  $m - (m - \tk) = \tk$.
  Therefore the number of elements in $\field_2^m/\rs{\h^t} \otimes \ker{\h^t}$ is $\tk^2$.
  On counting the operators, we see that there are indeed $k^2$ vectors of VV type and $\tk^2$ vectors of CC type, thus adding up to the correct number of logical operators.
\end{proof}

\textbf{Example:} consider the surface code generated using two copies of the repetition code $[3,1,3]$ as shown in fig.~(\ref{fig:threebythree}) below.
Its factor graph $\G$ runs vertically on the left and horizontally on the bottom.
Variable nodes have been colored blue and indexed by numerals whereas check nodes are green and indexed by letters.
The product of two nodes is represented as
\begin{center}
\begin{tikzpicture}
  \node at (-2.6,0) {$VV:$};
  \node[v] (v1) at (-2,0) {};
  \node () at (-1.5,0) {$\times$};
  \node[v] (v2) at (-1,0) {};
  \node () at (-0.5,0) {$\rightarrow$};
  \node[v] (v) at (0,0) {};
  \node[vv] (vv) at (0,0) {};

  \node at (-2.6,-1) {$CC:$};
  \node[c] (c1) at (-2,-1) {};
  \node () at (-1.5,-1) {$\times$};
  \node[c] (c2) at (-1,-1) {};
  \node () at (-0.5,-1) {$\rightarrow$};
  \node[c] (c) at (0,-1) {};
  \node[cc] (cc) at (0,-1) {};
\end{tikzpicture}
\qquad
\begin{tikzpicture}
  \node () at (-2.6,0) {$X:$};
  \node[v] (v1) at (-2,0) {};
  \node () at (-1.5,0) {$\times$};
  \node[c] (c2) at (-1,0) {};
  \node () at (-0.5,0) {$\rightarrow$};
  \node[c] (c) at (0,0) {};
  \node[vv] (vv) at (0,0) {};

  \node () at (-2.6,-1) {$Z:$};
  \node[c] (c1) at (-2,-1) {};
  \node () at (-1.5,-1) {$\times$};
  \node[v] (v2) at (-1,-1) {};
  \node () at (-0.5,-1) {$\rightarrow$};
  \node[v] (v) at (0,-1) {};
  \node[cc] (cc) at (0,-1) {};
\end{tikzpicture}
\end{center}
Mapping this to the surface code, we choose a convention where $Z$ stabilizers correspond to plaquettes and $X$ stabilizers correspond to vertices.

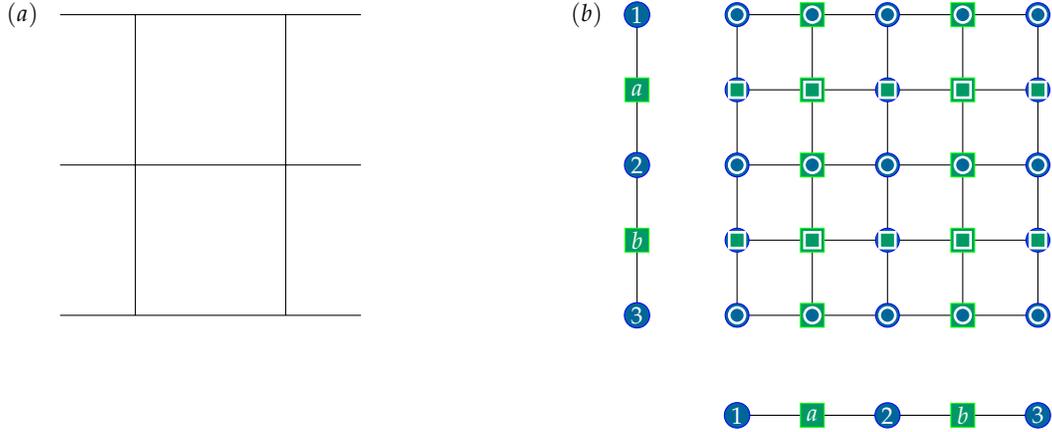
\begin{figure}[h]
  \centering
  \begin{tikzpicture}
    \node at (-9.5,4) {$(a)$};
    \draw[step=2cm] (-9,0) grid (-5,4);
    \node at (-2,4) {$(b)$};
    \draw[step=1cm] (0,0) grid (4,4);
    \foreach \x in {0,2,4}
        \foreach \y in {0,...,4}
           \node [v]  (v\x\y) at (\x,\y) {};

     \foreach \x in {1,3}
        \foreach \y in {0,...,4}
          \node [c]  (c\x\y) at (\x,\y) {};

     \foreach \x in {0,...,4}
         \foreach \y in {0,2,4}
           \node [vv]  (\x\y\x\y) at (\x,\y) {};

       \foreach \x in {0,...,4}
         \foreach \y in {1,3}
           \node [cc]  (\x\y\x\y) at (\x,\y) {};

       \draw (-1.33,0) edge[-] (-1.33,4);

         \node [v] (3) at (-1.33,0) {{\color{white}3}};
         \node [v] (2) at (-1.33,2) {{\color{white}2}};
         \node [v] (1) at (-1.33,4) {{\color{white}1}};

         \node [c] (B) at (-1.33,1) {{\color{white}$b$}};
         \node [c] (A) at (-1.33,3) {{\color{white}$a$}};

       \draw (0,-1.33) edge[-] (4,-1.33);

         \node [v] (1) at (0,-1.33) {{\color{white}1}};
         \node [v] (2) at (2,-1.33) {{\color{white}2}};
         \node [v] (3) at (4,-1.33) {{\color{white}3}};

         \node [c] (A) at (1,-1.33) {{\color{white}$a$}};
         \node [c] (B) at (3,-1.33) {{\color{white}$b$}};
  \end{tikzpicture}
  \caption{
  $(a)$ A surface code with smooth boundaries above and below and rough boundaries on the sides.
  $(b)$ The corresponding $3 \times 3$ hypergraph product code.}
  \label{fig:threebythree}
\end{figure}

To break it down further, we could consider each of the constituent parts of the definition above.
First, note that the qubits of VV type, depicted using two concentric circles, represent horizontal edges in the surface code whereas qubits of CC type, depicted using concentric squares, are vertical edges of the surface code.
The action of the $X$ stabilizers on qubits of VV type is defined by the matrix $\1[1,2,3] \otimes \h$ and is depicted in fig. \ref{fig:graphProductX}(a).
Similarly, the action of the $X$ stabilizers on qubits of CC type is defined by $\h^t \otimes \1[A,B]$ and is depicted in fig. \ref{fig:graphProductX}(b).
\begin{figure}[h]
  \centering
    \begin{tikzpicture}
    \node at (-2,4) {$(a)$};
    \draw (0,0) edge[-] (4,0);
    \draw (0,2) edge[-] (4,2);
    \draw (0,4) edge[-] (4,4);
    
    \foreach \x in {0,2,4}
        \foreach \y in {0,2,4} 
           \node [v]  (v\x\y) at (\x,\y) {};

   \foreach \x in {0,2,4}
       \foreach \y in {1,3} 
          \node [classv]  (v\x\y) at (\x,\y) {};

     \foreach \x in {1,3}
        \foreach \y in {0,2,4} 
          \node [c]  (c\x\y) at (\x,\y) {};

      \foreach \x in {1,3}
         \foreach \y in {1,3} 
           \node [classc]  (c\x\y) at (\x,\y) {};

     \foreach \x in {0,...,4}
         \foreach \y in {0,2,4} 
           \node [vv]  (\x\y\x\y) at (\x,\y) {};
           
     \foreach \x in {0,...,4}
         \foreach \y in {1,3} 
           \node [classcc]  (\x\y\x\y) at (\x,\y) {};
          

        \node [v] (3) at (-1.33,0) { {\color{white} 3}};
        \node [v] (2) at (-1.33,2) { {\color{white} 2}};
        \node [v] (1) at (-1.33,4) { {\color{white} 1}};
        
        \node [classc] (A) at (-1.33,3) {$a$};
        \node [classc] (A) at (-1.33,1) {$b$};
      
      \draw (0,-1.33) edge[-] (4,-1.33);
      
        \node [v] (1) at (0,-1.33) { {\color{white} 1}};
        \node [v] (2) at (2,-1.33) { {\color{white} 2}};
        \node [v] (3) at (4,-1.33) { {\color{white} 3}};
      
        \node [c] (A) at (1,-1.33) { {\color{white} $a$}};
        \node [c] (B) at (3,-1.33) { {\color{white} $b$}};
  \end{tikzpicture}
  \qquad\qquad
    \begin{tikzpicture}
      \node at (-2,4) {$(b)$};
      \draw (1,0) edge[-] (1,4);
      \draw (3,0) edge[-] (3,4);
         
       \foreach \x in {1,3}
          \foreach \y in {0,...,4}
            \node [c]  (c\x\y) at (\x,\y) {};
        
        \foreach \x in {0,2,4}
            \foreach \y in {0,...,4} 
               \node [classv]  (v\x\y) at (\x,\y) {};

       \foreach \x in {1,3}
           \foreach \y in {0,2,4}
             \node [vv]  (\x\y\x\y) at (\x,\y) {};

       \foreach \x in {0,2,4}
           \foreach \y in {0,2,4}
             \node [classvv]  (\x\y\x\y) at (\x,\y) {};

         \foreach \x in {1,3}
           \foreach \y in {1,3}
             \node [cc]  (\x\y\x\y) at (\x,\y) {};
        
         \foreach \x in {0,2,4}
             \foreach \y in {1,3} 
                \node [classcc]  (v\x\y) at (\x,\y) {};

         \draw (-1.33,0) edge[-] (-1.33,4);

           \node [v] (3) at (-1.33,0) {{\color{white}3}};
           \node [v] (2) at (-1.33,2) {{\color{white}2}};
           \node [v] (1) at (-1.33,4) {{\color{white}1}};

           \node [c] (B) at (-1.33,1) {{\color{white} $b$}};
           \node [c] (A) at (-1.33,3) {{\color{white} $a$}};

           
           \node [classv] (1) at (0,-1.33) {1};
           \node [classv] (2) at (2,-1.33) {2};
           \node [classv] (3) at (4,-1.33) {3};
           \node [c] (A) at (1,-1.33) {{\color{white} $a$}};
           \node [c] (B) at (3,-1.33) {{\color{white} $b$}};
    \end{tikzpicture}
 \caption{
  The action of $X$ stabilizers on VV qubits on the left and CC qubits on the right.
  Filled nodes represent nodes involved in the stabilizer generator.
 }
 \label{fig:graphProductX}
\end{figure}
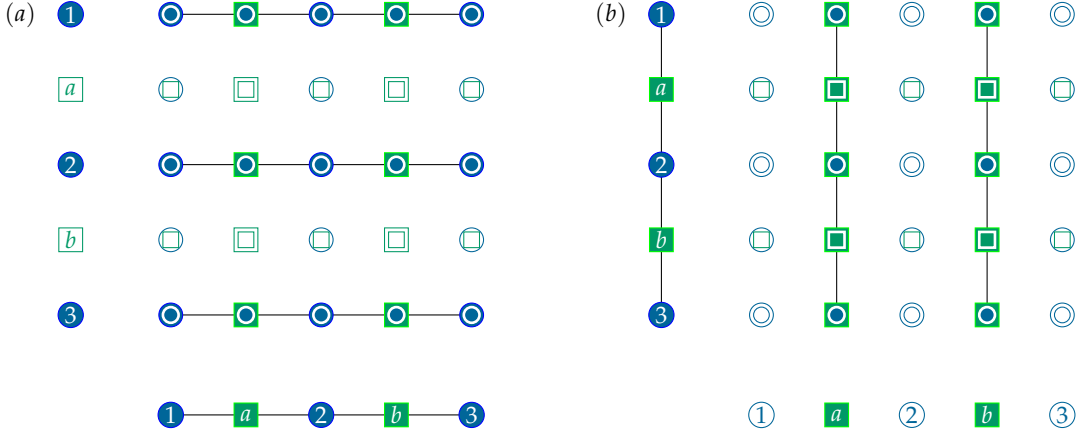

In a similar manner, we can obtain the representation for $Z$ stabilizers.
By overlaying these diagrams, we obtain fig. (\ref{fig:threebythree}).

\section{Punctures}
\label{sec:punctures}

A puncture is a defect on the hypergraph product created by removing both qubits and stabilizers belonging to some (small) portion of the code.
This shall be effected by measuring single-qubit Pauli operators within the interior of the puncture.
This is similar to creating a puncture on the surface code \cite{vuillot2019code}.

\subsection{Definition}
\label{subsec:defPunc}
We begin this section with some notation.
Let $S \subseteq V$ denote a connected subset of variable nodes.
$N = \nbhd(S) \subseteq C$ is its neighborhood and $A = \nbhd^{-1}(S)$ is its ancestor as shown in fig. \ref{fig:identities}(a).
\begin{align*}
  N = \{c \in C : \exists u \in S \text{ such that } (u,c) \in E\} \qquad A = \{c \in C: \forall u \in \Gamma(c), u \in S\}.
\end{align*}
For any set $V' \subseteq V$, we let $\1_{V'}$ denote the projector on $V'$ over $\field_2^{V}$.
We also write $\h_{V'} = \h\1_{V'}$ for the restriction of the parity check matrix to $V'$.

Similarly, let $T \subseteq C$  denote a connected subset of check nodes.
$M = \nbhd(T) \subseteq V$ is its neighborhood and $B = \nbhd^{-1}(T)$ is its inverse neighborhood as shown in fig. \ref{fig:identities}(b).
\begin{align*}
  M = \{v \in V : \exists c \in T \text{ such that } (v,c) \in E\} \qquad B = \{v \in V: \forall c \in \Gamma(v), c \in T\}.
\end{align*}
For any set $C' \subseteq C$, we let $\1_{C'}$ denote the projector on $C'$ over $\field_2^{C}$.
We also write $\h_{C'} = \1_{C'} \h$ for the restriction of the parity check matrix to $C'$.

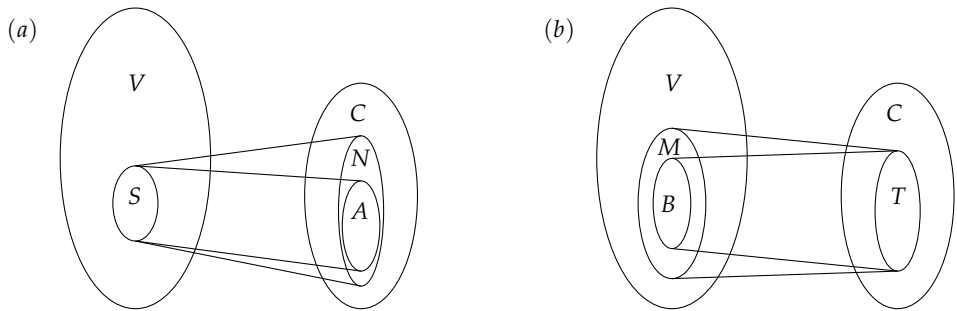
\begin{figure}[h]
  \centering
  \begin{tikzpicture}
      \node at (-1.5,1.7) {$(a)$};
      \draw (0,0) ellipse (1cm and 2cm);
      \draw (3,-0.5) ellipse (0.75cm and 1.5cm);
      \draw (0,-0.6) ellipse (0.3cm and 0.5cm);
      \draw (3,-0.9) ellipse (0.25cm and 0.6cm);
      \draw (3,-0.7) ellipse (0.3cm and 1cm);
      \draw (3,-1.5) edge[-] (0,-1.1) (3,-0.3) edge[-] (0,-0.1);
      \draw (0,-0.1) edge[-] (3,0.3) (0,-1.1) edge[-] (3,-1.7);
      \node[text width=1cm] at (3.35,0) {$N$};
      \node[text width=1cm] at (3.35,-0.7) {$A$};
      \node[text width=1cm] at (0.4,-0.5) {$S$};
      \node[text width=1cm] at (0.4,1) {$V$};
      \node[text width=1cm] at (3.35,0.6) {$C$};
    \end{tikzpicture}
    \qquad\qquad
    \begin{tikzpicture}
      \node at (-1.5,1.7) {$(b)$};
      \draw (0,0) ellipse (1cm and 2cm);
      \draw (3,-0.5) ellipse (0.75cm and 1.5cm);
      \draw (0,-0.6) ellipse (0.45cm and 1cm);
      \draw (0,-0.6) ellipse (0.25cm and 0.6cm);
      \draw (3,-0.7) ellipse (0.3cm and 0.8cm);
      \draw (0,-1.2) edge[-] (3,-1.5) (0,0) edge[-] (3,0.1);
      \draw (3,0.1) edge[-] (0,0.4) (3,-1.5) edge[-] (0,-1.6);
      \node[text width=1cm] at (0.3,0.15) {$M$};
      \node[text width=1cm] at (0.35,-0.6) {$B$};
      \node[text width=1cm] at (3.4,-0.5) {$T$};
      \node[text width=1cm] at (0.4,1) {$V$};
      \node[text width=1cm] at (3.35,0.6) {$C$};
    \end{tikzpicture}
  \caption{Schematic of factor graphs.
  $(a)$ Denotes the subgraph induced by $S$. $N$ is its neighborhood and $A$ is its ancestor.
  $(b)$ Denotes the subgraph induced by $T$. $M$ is its neighborhood and $B$ is its ancestor.}
  \label{fig:identities}
\end{figure}
At this juncture, we make some observations that will be useful later.
\begin{align*}
  \Gamma(A) \subseteq S \implies \Gamma(S^c) \subseteq A^c &\qquad \Gamma(B) \subseteq T \implies \Gamma(T^c) \subseteq B^c \\
  A^c = N^c \union (N\setminus A) &\qquad B^c = M^c \union (M \setminus B)~.
\end{align*}

To create a puncture on the quantum code, we will stop measuring certain stabilizers, and modify others when carving out a portion of the interior.
The punctures will be classified by how stabilizers are modified.

\begin{definition}[\textbf{Smooth puncture}]
\label{def:smooth}
Let $S \subseteq V$ and $T \subseteq C$ be connected sets of variable and check nodes.
Let $N, A$ and $M, B$ denote induced sets as defined above.
A smooth puncture is defined by the stabilizers $\h_X'$ and $\h_Z'$ where
\begin{align*}
  \h_X' = \left(\1_B \otimes \h_S | \h_T^t \otimes \1_A \right) \qquad \h_Z' = \left(\h_T \otimes \1_S | \1_T \otimes \h_S^t \right)~.
\end{align*}
\end{definition}
Note that this is not exactly the graph product of the two subgraphs selected by $T$ and $S$.
This is verified by noting that it is missing elements from $M \times S$.
Rather it follows the hypergraph product construction on the interior nodes of both graphs.
For simplicity, we abuse notation and refer to this as the graph product  $T \times S$.

The defining trait of a smooth puncture is that $Z$ stabilizers are not broken across its boundary.
We refer to the schematic in fig. \ref{fig:schematic}(a) below.
Such stabilizers would have to be of the form $(c,v)$ for some check $c \in C$ and $v \in V$ where either the check node $c$ or the variable node $v$ are in the boundary of $T$ or $S$ respectively.
This does not exist by construction -- check nodes in $T$ are contained entirely within the puncture, as are variable nodes in $S$.
The internal qubits of a smooth puncture are the nodes $B \times S \union T \times A$ and will be measured in the $Z$ basis to create the puncture.
The qubits on the boundary of a smooth puncture correspond to the sets
\begin{align*}
  (M\setminus B) \times S \union T \times (N \setminus A)~.
\end{align*}
The $X$ stabilizers on the boundary of a smooth puncture correspond to the sets
\[
  (M \setminus B) \times N \union M \times (N \setminus A)~.
\]
Their support on the interior of the puncture, $B \times S \union T \times A$, is removed.
Therefore $X$ stabilizers on the boundary of a smooth puncture are broken.

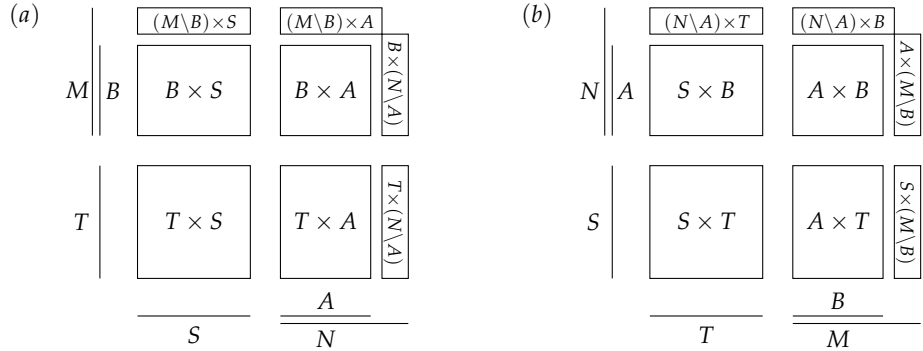
\begin{figure}[h]
  \centering
  \begin{tikzpicture}
    \draw (0,0) rectangle (1.5,1.5); 
    \node at (0.75,0.75) {$T \times S$};
    \draw (1.9,1.9) rectangle (3.1,3.1); 
    \node at (2.5,2.5) {$B \times A$};
    \draw (1.9,3.25) rectangle (3.25,3.6); 
    \node at (2.55,3.4) {{$\scriptstyle (M \setminus B) \times A$}};
    \draw (3.25,1.9) rectangle (3.6,3.25); 
    \node[rotate=-90] at (3.4,2.6) {{$\scriptstyle B \times (N\setminus A)$}};
    \draw (1.9,0) rectangle (3.1,1.5); 
    \node at (2.5,0.75) {$T \times A$};
    \draw (0,1.9) rectangle (1.5,3.1); 
    \node at (0.75,2.5) {$B \times S$};
    \draw (-0.5,0) edge[-] (-0.5,1.5); 
    \node at (-0.75,0.75) {$T$};
    \draw (-0.5,1.9) edge[-] (-0.5,3.1); 
    \node at (-0.325,2.5) {$B$};
    \draw (-0.6,1.9) edge[-] (-0.6,3.6); 
    \node at (-0.8,2.5) {$M$};
    \draw (0,-0.5) edge[-] (1.5,-0.5); 
    \node at (0.75,-0.75) {$S$};
    \draw (1.9,-0.5) edge[-] (3.1,-0.5); 
    \node at (2.5,-0.3) {$A$};
    \draw (1.9,-0.6) edge[-] (3.6,-0.6); 
    \node at (2.5,-0.8) {$N$};
    \draw (0,3.25) rectangle (1.5,3.6); 
    \node at (0.75,3.4) {{$\scriptstyle (M \setminus B) \times S$}};
    \draw (3.25,0) rectangle (3.6,1.5); 
    \node[rotate=-90] at (3.4,0.75) {{$\scriptstyle T \times (N\setminus A)$}};
    \node at (-1.5,3.5) {$(a)$};
  \end{tikzpicture}
  \qquad
  \qquad
  \begin{tikzpicture}
    \draw (0,0) rectangle (1.5,1.5); 
    \node at (0.75,0.75) {$S \times T$};
    \draw (1.9,1.9) rectangle (3.1,3.1); 
    \node at (2.5,2.5) {$A \times B$};
    \draw (1.9,3.25) rectangle (3.25,3.6); 
    \node at (2.55,3.4) {{$\scriptstyle (N \setminus A) \times B$}};
    \draw (3.25,1.9) rectangle (3.6,3.25); 
    \node[rotate=-90] at (3.4,2.6) {{$\scriptstyle A \times (M \setminus B)$}};
    \draw (1.9,0) rectangle (3.1,1.5); 
    \node at (2.5,0.75) {$A \times T$};
    \draw (0,1.9) rectangle (1.5,3.1); 
    \node at (0.75,2.5) {$S \times B$};
    \draw (-0.5,0) edge[-] (-0.5,1.5); 
    \node at (-0.75,0.75) {$S$};
    \draw (-0.5,1.9) edge[-] (-0.5,3.1); 
    \node at (-0.325,2.5) {$A$};
    \draw (-0.6,1.9) edge[-] (-0.6,3.6); 
    \node at (-0.8,2.5) {$N$};
    \draw (0,-0.5) edge[-] (1.5,-0.5); 
    \node at (0.75,-0.75) {$T$};
    \draw (1.9,-0.5) edge[-] (3.1,-0.5); 
    \node at (2.5,-0.3) {$B$};
    \draw (1.9,-0.6) edge[-] (3.6,-0.6); 
    \node at (2.5,-0.8) {$M$};
    \draw (0,3.25) rectangle (1.5,3.6); 
    \node at (0.75,3.4) {{$\scriptstyle (N \setminus A) \times T$}};
    \draw (3.25,0) rectangle (3.6,1.5); 
    \node[rotate=-90] at (3.4,0.75) {{$\scriptstyle S \times (M\setminus B)$}};
    \node at (-1.5,3.5) {$(b)$};
  \end{tikzpicture}
  \caption{Schematic for the puncture.
  The subgraphs selected by $T$ and $S$ are flattened and placed below and to the left.
  Their product is represented using four quadrants.
  The qubits are in the North-West and South-East quadrants.
  The $Z$ stabilizers are in the South-West quadrant.
  The $X$ stabilizers are in the North-East quadrant.
  (a) Smooth puncture defined by $T$ and $S$.
  The $Z$ stabilizers $T\times S$ are completely within the puncture and are thus not broken.
  (b) Rough puncture defined by $S$ and $T$.
  The $X$ stabilizers $S\times T$ are completely within the puncture and are thus not broken.}
  \label{fig:schematic}
\end{figure}

In a similar manner, a rough puncture can be created by interchanging the roles of $T$ and $S$ on the graphs.
It is formally defined as follows.
\begin{definition}[\textbf{Rough puncture}]
\label{def:rough}
Let $S \subseteq V$ and $T \subseteq C$ be connected sets of variable and check nodes.
Let $N, A$ and $M, B$ denote induced sets as defined above.
A rough puncture is defined by the stabilizers $\h_X'$ and $\h_Z' where
$\begin{align*}
  \h_X' = \left(\1_S \otimes \h_T | \h_S^t \otimes \1_T \right) \qquad \h_Z' = \left(\h_S \otimes \1_B | \1_A \otimes \h_T^t \right)~.
\end{align*}
\end{definition}
Abusing notation, this can be thought of as a graph product $S \times T$.

The defining trait of a rough puncture is that $X$ stabilizers are not broken across the boundary.
We refer to the schematic in fig. \ref{fig:schematic} (b) for the following discussion.
Such stabilizers would have to be of the form $(v,c)$ for some check $c \in C$ and $v \in V$ where either the check node $c$ or the variable node $v$ are in the boundary of $S$ or $T$ respectively.
For the same reasons as before, such nodes do not exist.
The internal qubits of a rough puncture are the nodes $S \times B \union A \times T$ and will be measured in the $X$ basis to create the puncture.
The qubits on the boundary of a rough puncture correspond to the sets
\begin{align*}
  S \times (M\setminus B) \union (N \setminus A) \times T~.
\end{align*}
The $Z$ stabilizers on the boundary of a rough puncture correspond to the sets
\[
  N \times (M \setminus B) \union (N \setminus A) \times M~.
\]
Their support on the interior of the puncture, $S \times B \union A \times T$ is removed.
Hence $Z$ stabilizers on the boundary of a rough puncture are broken.

To deform $\Q$, we remove the edges that are contained in a puncture.
Algebraically, it is described by $\h_X + \h_X'$ and $\h_Z + \h_Z'$.
This code is itself not a hypergraph product code but is clearly LDPC.
We are merely puncturing an LDPC code; by removing edges, we cannot increase the weight of checks.

We first show that the code defined this way obeys the desired commutation relations.
Before doing so, it is useful to note the following identity:
\begin{align}
\label{eq:identity}
  \h_B = \h \1_B = \1_T \h \1_B \qquad \h_A = \1_A \h = \1_A \h \1_S~.
\end{align}
This follows from the fact that the neighborhoods of sets $B$ and $A$ are completely contained within the sets $T$ and $S$ by definition.

\begin{lemma}
  \label{lem:validCommutation}
  The punctured code forms a valid stabilizer code.
\end{lemma}
\begin{proof}
 Consider a smooth puncture created by two subsets $T \subseteq C$ and $S\subseteq V$.
 The case of a rough puncture follows similarly.
 The punctured code has stabilizers
 \begin{align*}
   \h_X + \h_X' \qquad \h_Z + \h_Z'~.
 \end{align*}
 We already know that $\h_X \h_Z^t = 0 \pmod{2}$.
 We need to check the other relations.
 \begin{align*}
  \h_X (\h_Z')^t &= (\1_{V} \otimes \h | \h^t \otimes \1_{C}) \left[(\h_{T} \otimes \1_{S} | \1_{T} \otimes \h_{S}^t)\right]^t\\
                 &= \h_T^t \otimes \h_S + \h_T^t \otimes \h_S  = 0 \pmod{2} \\
   \h_X' (\h_Z)^t &= (\1_{B} \otimes \h_{S} | \h_{T}^t \otimes \1_{A}) \left[(\h \otimes \1_V | \1_C \otimes \h^t)\right]^t\\
                  &=  \h_{B}^t \otimes \h_S + \h_T^t \otimes \h_A\\
   \h_X' (\h_Z')^t &= (\1_{B} \otimes \h_{S} | \h_{T}^t \otimes \1_{A})\left[(\h_{T} \otimes \1_S | \1_T \otimes \h_S^t)\right]^t\\
                   &= \h_B^t \otimes \h_S + \h_T^t \otimes \h_A~.
 \end{align*}
 In the last line, we have used the identity in eq. \ref{eq:identity}.
 Inspecting the last two equations, we find that each term appears twice.
 Therefore the sum of all the terms in these two equations is $0$ mod $2$ as desired.
\end{proof}

For convenience, we have summarized this section in table \ref{tab:summary}.
\begin{table}[h]
\def\arraystretch{2}
\begin{tabular}{|c|c|c|}
  \hline
          & Smooth puncture                                         & Rough puncture\\
  \hline
  $\h_X'$ & $\left(\1_B \otimes \h_S | \h_T^t \otimes \1_A \right)$ & $\left(\1_S \otimes \h_T | \h_S^t \otimes \1_T \right)$\\
  \hline
  $\h_Z'$ & $\left(\h_T \otimes \1_S | \1_T \otimes \h_S^t \right)$ & $\left(\h_S \otimes \1_B | \1_A \otimes \h_T^t \right)$\\
  \hline
  Internal qubits & $(B \times S) \union (T \times A)$ & $(S \times B) \union (A \times T)$\\
  \hline
  Boundary qubits & $(M\setminus B)\times S \union T\times (N \setminus A)$ & $S\times (M\setminus B)\union(N \setminus A)\times T$ \\
  \hline
  Boundary $X$ stabilizers & $(M \setminus B) \times N \union M \times (N \setminus A)$ & $\emptyset$ \\
  \hline
  Boundary $Z$ stabilizers & $\emptyset$ & $N \times (M \setminus B) \union (N \setminus A) \times M$\\
  \hline
\end{tabular}
\caption{Summary of properties of punctures. We assume that $S \subseteq V$ and $T \subseteq C$ are (connected) subsets of variable and check nodes.
$S$ induces the sets $N$ and $A$, its neighborhood and ancestor, and similarly $T$ induces the sets $M$ and $B$.}
\label{tab:summary}
\end{table}

\subsection{Example: Surface code}
We illustrate these ideas with the surface code.
This code is formed using two $[5,1,5]$ repetition codes.
We choose the sets $T = \{b,c\}$ and $S = \{3,4\}$.
For the sake of completeness, all the neighbors and ancestors are listed below.
These subsets, along with the quantum code they generate, are shown in fig.~(\ref{fig:PQ}).
\begin{align*}
  \begin{matrix}
    T = \{b,c\} & M = \{2,3,4\} &  B = \{3\}   \\
    S = \{3,4\}    & N = \{b,c,d\} & A = \{c\} 
  \end{matrix}~.
\end{align*}

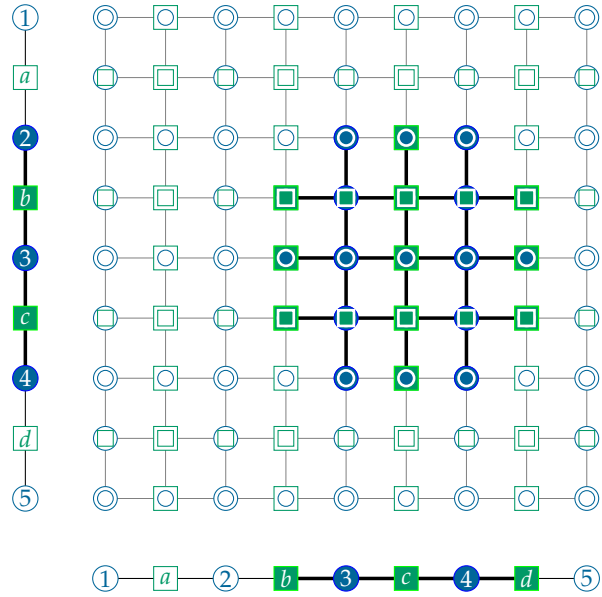
\begin{figure}[h]
  \centering
  \begin{tikzpicture}[scale=0.8]
    \draw[step=2cm,gray,very thin] (0,0) grid (8,8);
    \draw[step=2cm,gray,very thin,xshift=1cm,yshift=1cm] (-0.9,-1) grid (7,7);
    \draw (4,2) edge[-,line width=1.33pt] (4,6);
    \draw (5,2) edge[-,line width=1.33pt] (5,6);
    \draw (6,2) edge[-,line width=1.33pt] (6,6);
    
    \draw (3,3) edge[-,line width=1.33pt] (7,3);
    \draw (3,4) edge[-,line width=1.33pt] (7,4);
    \draw (3,5) edge[-,line width=1.33pt] (7,5);

    \foreach \x in {0,2,4,6,8}
        \foreach \y in {0,...,8}
           \node [classv]  (v\x\y) at (\x,\y) {};

     \foreach \x in {1,3,5,7}
        \foreach \y in {0,...,8} 
          \node [classc]  (c\x\y) at (\x,\y) {};

     \foreach \x in {0,...,8}
         \foreach \y in {0,2,4,6,8} 
           \node [classvv]  (\x\y\x\y) at (\x,\y) {};

       \foreach \x in {0,...,8}
         \foreach \y in {1,3,5,7} 
           \node [classcc]  (\x\y\x\y) at (\x,\y) {};
           
     \foreach \x in {4,6}
         \foreach \y in {2,3,4,5,6}
            \node [v]  (v\x\y) at (\x,\y) {};

      \foreach \x in {3,5,7}
         \foreach \y in {3,4,5} 
           \node [c]  (c\x\y) at (\x,\y) {};

     \foreach \x in {5}
        \foreach \y in {2,6} 
          \node [c]  (c\x\y) at (\x,\y) {};

      \foreach \x in {3,4,5,6,7}
          \foreach \y in {3,4,5} 
            \node [vv]  (\x\y\x\y) at (\x,\y) {};
      
      \foreach \x in {4,5,6}
          \foreach \y in {2,6} 
            \node [vv]  (\x\y\x\y) at (\x,\y) {};

        \foreach \x in {3,4,5,6,7}
          \foreach \y in {3,5} 
            \node [cc]  (\x\y\x\y) at (\x,\y) {};

     \draw (-1.33,0) edge[-] (-1.33,8);
     \draw (-1.33,2) edge[-,line width=1.33pt] (-1.33,6);
      
       \node[classv] (5) at (-1.33,0) {5};
       \node[v] (4) at (-1.33,2) {{\color{white} 4}};
       \node [v] (3) at (-1.33,4) { {\color{white} 3}};
       \node[v] (2) at (-1.33,6) { {\color{white} 2}};
       \node [classv] (1) at (-1.33,8) { 1};
       
      \node[classc] (D) at (-1.33,1) { $d$};
       \node[c] (C) at (-1.33,3) { {\color{white} $c$}};
       \node [c] (B) at (-1.33,5) { {\color{white}  $b$}};
       \node [classc] (A) at (-1.33,7) { $a$};
           
       \draw (0,-1.33) edge[-] (8,-1.33);
       \draw (3,-1.33) edge[-,line width=1.33pt] (7,-1.33);
       
         \node [classv] (1) at (0,-1.33) { 1};
         \node [classv] (2) at (2,-1.33) { 2};
         \node [v] (3) at (4,-1.33) { {\color{white} 3}};
         \node [v] (4) at (6,-1.33) { {\color{white} 4}};
         \node [classv] (5) at (8,-1.33) { 5};
       
         \node [classc] (A) at (1,-1.33) { $a$};
         \node [c] (B) at (3,-1.33) { {\color{white} $b$}};
         \node [c] (C) at (5,-1.33) { {\color{white} $c$}};
         \node [c] (D) at (7,-1.33) { {\color{white} $d$}};
  \end{tikzpicture}
  \caption{Subcode created by $T \times S$.
           Shaded nodes are part of the puncture.}
  \label{fig:PQ}
\end{figure}

In the algebraic description, we let $\h_T$ and $\h_S$ be the matrices correspond to the subcodes corresponding to the subsets selected above.
\begin{align*}
 \h_T = \1\{b,c\} \h = \begin{pmatrix}
                         0& 0& 0& 0&0  \\
                         0& 1& 1& 0&0  \\
                         0& 0& 1& 1&0  \\
                         0& 0& 0& 0&0
                       \end{pmatrix}
       \qquad
\h_S = \h \1\{3,4\} = \begin{pmatrix}
                       0& 0& 0& 0& 0\\
                       0& 0& 1& 0& 0 \\
                       0& 0& 1& 1& 0\\
                       0& 0& 0& 1&0
                     \end{pmatrix}
\end{align*}

The $X$ and $Z$ stabilizers that are removed from the puncture can then be described as
\begin{align*}
\h_X' = \left(\1_B \otimes \h_S | \h_T^t \otimes \1_A \right) \qquad \h_Z' = \left(\h_T \otimes \1_S | \1_T \otimes \h_S^{t} \right)~.
\end{align*}

\subsection{Logical Pauli operators for punctures}
\label{subsec:logicalPaulis}

Mirroring the surface code, punctured hypergraph product codes support two types of logical operators - {\em loop-type operators} that exist only on the boundary of the puncture and {\em chain-type operators} that are supported on the boundary of the puncture and also extend into the rest of the code.
For the rest of this section, we let $T \subseteq C$ and $S \subseteq V$ be some connected subsets.
The sets $M \subseteq V$, $N \subseteq C$ are the respective neighborhoods, and $B \subseteq V$, $A \subseteq C$ are the respective ancestors.
We shall derive the form of the logical operators for a smooth puncture defined as above.
The logical operators of a rough puncture will follow by exchanging the roles of $S$ and $T$.

Before proceeding, we impose certain constraints on how $S$ and $T$ are chosen.
These constraints apply to smooth and rough punctures both.
The constraints stipulate that certain subcodes associated with $S$ and $T$ are \emph{correctable}, i.e. if these portions of the code were erased then we do not lose any codewords of the underlying code in this process.
    
\begin{definition}[Correctability]
  \label{def:correctable}
  A puncture defined by $T \subseteq C$ and $S \subseteq V$ is \emph{correctable} if:
  \begin{enumerate}
    \item \label{it:interior} $\ker{\h_N} = \ker{\h_M^t} = \emptyset$.
    \item \label{it:interior2} $\ker{\h_T} = \ker{\h_S^t} = \emptyset$.
  \end{enumerate}
\end{definition}

These conditions imply that certain subsets of the parity check matrix also have trivial kernels.
\begin{lemma}
  \label{lem:correctable-implication}
  The correctability condition implies the following relations:
  \begin{enumerate}
    \item $\ker{\h_S} = \ker{\h_T^t} = \emptyset$.
    \item $\ker{\h_A^t} = \ker{\h_B} = \emptyset$.
  \end{enumerate}
\end{lemma}
\begin{proof}
  Both claims follow identical proofs.
  We show one of them here and then outline the proof for the rest.
  Observe that we may write $\h_N$ as the sum of matrices supported only on $S$ and $S^c$ respectively:
  \begin{align}
    \h_N = \1_N\h = \h\1_S + \1_N\h\1_{S^c}~.
  \end{align}
  If there existed an element $\beta \in \field_2^V$ such that $\supp{\beta} \in S$, and $\beta \in \ker{\h_S}$, then this would violate the condition that $\ker{\h_N}$ is empty.

  Similarly, the other claims follow once we note the identities listed below:
  \begin{align*}
    \h_M^t &= \1_M\h^t = \h^t\1_T + \1_M\h\1_{T^c}\\
    \h_T &= \1_T\h = \h\1_B + \1_T \h \1_{B^c}\\
    \h_S^t &= \1_S\h^t = \h^t \1_A + \1_S \h \1_{A^c}~.
  \end{align*}
  This completes the proof.
\end{proof}

Henceforth we shall only consider punctures that are correctable.
To foreshadow the next few results, we will argue that this condition can be used together with the cleaning lemma \cite{bravyi2009no} to guarantee that the embedded logical operators of the quantum code are unaffected by the puncture.

\subsubsection{Logical Z operators}

We first establish that creating a puncture will not affect the embedded logical $Z$ operators.
\begin{lemma}
  \label{lem:interiorClean}
  There are no embedded $Z$ logical operators supported within the interior of the puncture.
\end{lemma}
\begin{proof}
  The proof idea is to show that if a logical operator were completely contained within the puncture, then we would violate condition \ref{def:correctable}, part \ref{it:interior}.
  This will entail projecting down to the level of the classical code until we arrive at a contradiction.

  Let $\alpha \in \L_Z$ be a logical $Z$ operator , i.e. as given by lemma \ref{lem:embeddedLogicals}
  \begin{align}
    \alpha \in (\field_2^n/\rs{\h} \otimes \ker{\h}| \bz_{m^2} ) \union
               (\bz_{n^2}| \ker{\h^t} \otimes \field_2^m/\rs{\h^t})~.
  \end{align}
  For the sake of contradiction, let $\alpha$ be supported entirely within the interior of the puncture, i.e. only on $B \times S \union T \times A$.

  Without loss of generality, suppose the VV part of $\alpha$ is non-trivial.
  Let $\pi(\alpha)$ denote the projection of $\alpha$ on to the VV qubits.
  Let us index the elements of $\pi(\alpha)$ using variable nodes $u,v \in V$ as $\pi(\alpha)[u,v]$.
  Furthermore, let $\pi(\alpha)[u,*]$ denote the vector obtained by fixing the first component to $u$.
  Since $\pi(\alpha)$ is non-trivial and supported on $B \times S$, there must exist at least one $u \in B$ such that the vector $\beta_u := \pi(\alpha)[u,*]$ is non-trivial.

  Furthermore, it also implies that $\beta_u$ is supported only on $S$.
  However, this in turn implies that there exists a non-trivial element in $\ker{\h_S}$.

  As shown in lemma \ref{lem:correctable-implication}, this violates condition \ref{def:correctable} since if $\ker{\h_N}$ is empty, then so is $\ker{\h_S}$
  Therefore there cannot be any logical $Z$ operators completely contained within the interior.
\end{proof}

The next lemma will be useful in showing that operators in the row-space of $\h_Z'$ are not in the span of the $Z$ stabilizers $\h_Z + \h_Z'$, i.e. they are not redundant.
These operators will later be used to construct logical operators.
\begin{lemma}
  \label{lem:indpt}
  The stabilizers $\h_Z + \h_Z'$ outside the puncture are independent of the stabilizers $\h_Z'$ within the puncture, i.e.
  \begin{align*}
    \rs{\h_Z + \h_Z'} \cap \rs{\h_Z'} = \emptyset~.
  \end{align*}
\end{lemma}
\begin{proof}
  The proof idea is to study the overlap of the interior and exterior, and show that if there was a vector in the common support then we would violate condition \ref{def:correctable}, part \ref{it:interior}.

  For the sake of contradiction, suppose there exist vectors $a, b \in \field_2^{C \times V}$ such that
  \begin{align}
  \label{eq:assumption}
    a(\h_Z + \h_Z') = b\h_Z'~.
  \end{align}
  The stabilizers $\h_Z'$ in the interior and $\h_Z + \h_Z'$ in the exterior only share support along the boundary $(M \setminus B) \times S \union T \times (N \setminus A)$.

  Projecting both sides of eq. \ref{eq:assumption} on to the sets $M \times S \union T \times N$, i.e. the sets containing the interior and the boundary, we get
  \begin{align}
    a(\1_{T^c}\h\1_{M\setminus B} \otimes \1_S| \1_T \otimes \1_{S^c}\h^t\1_{N\setminus A}) = b(\h_T \otimes \1_S| \1_T \otimes \h_S^t)~.
  \end{align}
  By rearranging terms, we can find some non-trivial vector $c \in \field_2^{C\times V}$ such that this can be expressed as
  \begin{align}
    c(\h_M \otimes \1_S | \1_T \otimes \h_N^t) = 0~.
  \end{align}
  If $c$ is non-trivial, this requires that $\ker{\h_M^t}$ and $\ker{\h_N}$ are non-empty.
  However this violates condition \ref{def:correctable}, part \ref{it:interior}.
  This shows that $\h_Z + \h_Z'$ and $\h_Z'$ cannot have non-trivial overlap.
\end{proof}

With these conditions, we can study the logical $Z$ operators that emerge by creating a puncture.
These operators can be classified in terms of the (classical) codespaces associated with $\h_A$ and $\h_B^t$.
\begin{theorem}
  \label{thm:logicalZ}
  Let $\tG_B$ and $\g_A$ be the generator matrices for the codespaces defined by $\h_B^t$ and $\h_A$ respectively, i.e.
  the rows of $\tG_B$ and $\g_A$ span $\ker{\h_B^t}$ and $\ker{\h_A}$ respectively.
  The logical $Z$ operators are spanned by
  \begin{align*}
    (\tG_B^t \otimes \g_A^t) \h_Z'~.
  \end{align*}
\end{theorem}
\begin{proof}
  We begin from first principles.
  The stabilizers are given by $\h_X + \h_X'$ and the single-qubit operators in the interior of the puncture are described by the matrix $\INT = \left(\1_B \otimes \1_S | \1_T \otimes \1_A\right)$.
  The logical operators are defined as $\ker{\h_X + \h_X' + \INT}/\rs{\h_Z + \h_Z'}$.

  \textbf{Part 1:} $\ker{\h_X + \h_X'+\INT}$

  Suppose $\alpha \in \field_2^{n^2 + m^2}$ such that $\alpha \in \ker{\h_X + \h_X'+ \INT}$.
  We can assume that $\alpha$ is not supported in the interior, and consider the kernel of $\h_X + \h_X'$ instead of $\h_X + \h_X'+\INT$.
  We use lemma \ref{lem:interiorClean} together with the cleaning lemma \cite{bravyi2009no} to note that the embedded logicals are unaffected by the puncture.
  Any other $Z$ type operator that is supported in the interior will anti-commute with the single-qubit $X$ measurements used to generate the puncture and therefore will be removed.

  The operator $\alpha$ must therefore lie in the kernel of $\h_X$ outside the puncture.
  This contains the embedded logical operators and products of old stabilizers that were not in the interior.
  We shall only focus on the latter here in order to obtain the new logical operators.

  The stabilizers in the interior are spanned by $\h_Z'$.
  Those operators in $\rs{\h_Z'}$ but not supported in the interior are thus what we seek.
  The interior of the puncture corresponds to $B \times S \union T \times A$.
  Let $a \in \field_2^{C \times V}$ such that $\supp{a} \subseteq T \times S$.
  We want the projection of $a \h_Z'$ to vanish in the interior, i.e.
  \begin{align}
    a \h_Z' \1_{B \times S \union T \times A} &= 0\\
    a (\h_B \otimes \1_S | \1_T \otimes \h_A^t) &= 0~.
  \end{align}
  Inspecting the VV and CC parts of this equation separately, we find that we must have
  \begin{align}
    a^t \in \ker{\h_B^t} \otimes \ker{\h_A}~.
  \end{align}
  Equivalently, the space we desire is spanned by
  \begin{align}
  \label{eq:logicalZ}
    (\tG_B^t \otimes \g_A)\h_Z'~.
  \end{align}

  \textbf{Part 2:} $\rs{\h_Z + \h_Z'}$

  We refer to lemma \ref{lem:indpt} which states that the span of the stabilizers from within the puncture are independent of those outside the puncture.
  Therefore the space defined by eq. \ref{eq:logicalZ} is not in the span of the $Z$ stabilizers.
\end{proof}

\subsubsection{Logical X operators}

We now discuss the logical $X$ operators associated to a smooth puncture.

\begin{lemma}
  \label{lem:embeddedX}
  There are no embedded logical $X$ operators within the puncture.
\end{lemma}
\begin{proof}
  The proof idea is to show that if a logical $X$ operator were contained entirely within the puncture, then it violates the assumptions that the interior is correctable.

  Let $\alpha$ be a logical $X$ operator, i.e. as given by lemma \ref{lem:embeddedLogicals}
  \begin{align*}
    \alpha \in (\ker{\h} \otimes \field_2^n/\rs{\h} | \bz_{m^2}) \union (\bz_{n^2} | \field_2^{m}/\rs{\h^t} \otimes \ker{\h^t})~,
  \end{align*}
  which for the sake of contradiction is contained entirely within the interior $B \times S \union T \times A$.
  Without loss of generality, let us assume that the VV part of $\alpha$ is non-trivial.
  Let $\pi(\alpha)$ denote the projection of $\alpha$ on to the VV qubits.
  Let us index the elements of $\pi(\alpha)$ using variable nodes $u,v \in V$.
  It follows that $\pi(\alpha)$ is supported entirely on $B \times S$.

  Let $\pi(\alpha)[*,v]$ denote the vector obtained by fixing the second component to $v$.
  Since $\pi(\alpha)$ is non-trivial, there must exist at least one $v \in S$ such that the vector $\beta_v := \pi(\alpha)[*,v] \in \ker{\h}$.

  By assumption, since $\alpha$ is supported entirely on the interior, $\beta_v$ is supported only on $B$.
  However this in turn implies that there exists a non-trivial element in $\ker{\h_B}$.

  As shown in lemma \ref{lem:correctable-implication}, this violates condition \ref{def:correctable}, part \ref{it:interior2}.
  Therefore the puncture cannot contain any logical $X$ operators.
\end{proof}

The next lemma will help show that certain $X$ operators are not in the span of the stabilizer $\h_X + \h_X'$.
These operators will then be used to construct logical $X$ operators.
These are comprised of codewords of the classical codes that are complementary to the subgraphs chosen by $S$ and $T$.
In other words, they will involve the terms $\ker{\h_{T^c}}$ and $\ker{\h_{S^c}^t}$.

In the proof that follows, we shall make certain claims on these spaces.
Note that since the neighborhood the set $B$ is contained in the set $T$, it implies that the neighborhood of $T^c$ is contained within $B^c$.
Similarly, the neighborhood of $S^c$ is contained within $A^c$.
Therefore when studying $\ker{\h_{T^c}}$ and $\ker{\h_{S^c}^t}$, we shall assume that their support is contained in $B^c$ and $A^c$ respectively.

\begin{lemma}
  \label{lem:indpt2}
  Let $\alpha \in \field_2^{n^2 + m^2}$ such that $\supp{\alpha} \cap (M\setminus B) \times S \union T \times (N\setminus A) \neq \emptyset$ and it lies in one of the two following sets:
  \begin{enumerate}
    \item $\left(\ker{\h_{T^c}} \otimes \field_2^S / \rs{\h_A} | \bz_{m^2}\right)$; or
    \item $\left(\bz_{n^2}| \field_2^T / \rs{\h_B^t} \otimes \ker{\h_{S^c}^t} \right)$~.
  \end{enumerate}
  Then $\alpha$ does not lie in the row-span of $\h_X + \h_X'$.
\end{lemma}
\begin{proof}
  The proof idea is to successively project down to the level of the constituent classical codes until we arrive at a contradiction.

  We shall focus on the first object,
  \[
    \left(\ker{\h_{T^c}} \otimes \field_2^S / \rs{\h_A} | \bz_{m^2}\right)
  \]
  and note the other follows identically.

  For the sake of contradiction, let $\alpha \in (\ker{\h_{T^c}} \otimes \field_2^S / \rs{\h_A} | \bz_{m^2})$ such that it is in the row-span of $\h_X + \h_X'$.
  It follows that $\alpha$ has non-trivial support on $(M \setminus B) \times S$.

  By assumption, there exists a vector $a \in \field_2^{V \times C}$ such that
  \begin{align}
  \label{eq:contra1}
    a(\h_X + \h_X') = \alpha~.
  \end{align}
  Let us index $a$ as $a[v,c]$ for $v \in V$ and $c \in C$.
  Similarly, let $\pi(\alpha)$ denote the projection of $\alpha$ on to its VV part and let us index its elements as $\pi(\alpha)[u,v]$ for $u,v \in V$.

  We shall study the VV part first, and this will help us understand the support of the co-ordinate vector $a$.
  We wish to show that
  \begin{align}
  \label{eq:req}
    \supp{a} \cap (M\setminus B) \times A^c \neq \emptyset~.
  \end{align}
  As noted above, $\pi(\alpha)$ is supported on $(M\setminus B) \times S$.
  Therefore, if we had
  \begin{align*}
    \supp{a} \subseteq M^c \times A^c~,
  \end{align*}
  then for all $v \in S$, $\pi(\alpha)[*,v]$ cannot have support on $M \setminus B$.
  This follows from the fact that the VV portion of \ref{eq:contra1} is
  \begin{align*}
    a(\1_V \otimes \h + \1_B \otimes \h_S) = \pi(\alpha)~.
  \end{align*}
  In turn it would not be a part of $\ker{\h_{T^c}}$ and this violates the assumption on $\alpha$.
  For $u \in B^c$, $\pi(\alpha)[u,*]$ is supported on $\field_2^S/\rs{\h_A}$ and therefore $a[u,*]$ must be supported on $A^c$.
  This implies eq. \ref{eq:req}.

  With this condition on the support of $\alpha$ obtained from the VV side, we shall now show that we run into problems on the CC side.
  Since the CC portion of $\alpha$ is trivial,
  \begin{align}
    \label{eq:contra2}
    a(\h^t \otimes \1_C + \h_T^t \otimes \1_A) = 0~.
  \end{align}
  We may project this from the right on to the set $T \times C$ to obtain; it is on this set that we will run into a contradiction.
  After projection, this equation becomes
  \begin{align}
  \label{eq:contra3}
    a(\h_T^t \otimes \1_{A^c}) = 0~.
  \end{align}
  For some $c \in A^c$, we must have $b_c := a[*,c]$ such that $b_c$ has non-trivial support on $M\setminus B$.

  The vector $b$ has non-trivial support on $M \setminus B$ as explained, and therefore eq. \ref{eq:contra3} implies that
  \begin{align*}
    b_c \h_T^t = 0~.
  \end{align*}
  However this is not possible since $\ker{\h_T}$ is empty, as stated in condition \ref{def:correctable}, part \ref{it:interior2}.
  This completes the proof.
\end{proof}

These operators can be classified in terms of the (classical) codespaces associated with $\h_{T^c}$ and $\h_{S^c}^t$.
\begin{theorem}
\label{thm:logicalX}
  Let $O_Z$ be the $Z$ operators defined by
  \begin{align*}
    O_Z = (\h_M \otimes \1_S | \1_{T} \otimes \h_N^t)~,
  \end{align*}
  and let $\Omega_X = \ker{\h_Z + O_Z}$ denote the $X$ type operators in its kernel.

  The logical $X$ operators are described by
  \begin{align}
    \left[(\ker{\h_{T^c}} \otimes (\field_2^S / \rs{\h_A}) | \bz_{m^2}) \union
    (\bz_{n^2} | (\field_2^T / \rs{\h_B^t}) \otimes \ker{\h_{S^c}^t})\right]/\Omega_X~.
  \end{align}
\end{theorem}
Before proceeding to the proof, we make the following observations and highlight important features of this claim.
At first glance, the logical operators appear to break into two types, the VV type logicals defined by $\ker{\h_{T^c}}$ and the CC type operators defined by $\ker{\h_{S^c}^t}$.
Thus the logical $X$ operators are defined by the code spaces that are left over after the portions corresponding to $T$ and $S$ have been carved out.

The set $\h_Z + O_Z$ represents $Z$ stabilizers \emph{outside} the puncture.
Vectors in the kernel of $\h_Z + O_Z$ are unaffected by the addition of the puncture, and in that sense represent some invariant space.
The space $\Omega_X$ thus contains $X$ stabilizers and logicals whose support does not overlap with the puncture.
To help this object seem less alien, let us return to the surface code and consider an example.

Consider a smooth puncture defined on a surface code with only smooth boundaries.
This could define a logical qubit with the $X$ string running from the boundary of the smooth puncture to one of the boundaries of the lattice.
Depending on the arrangement, this logical operator could be supported only on CC qubits or only VV qubits as shown in fig. \ref{fig:smooth-everywhere}.
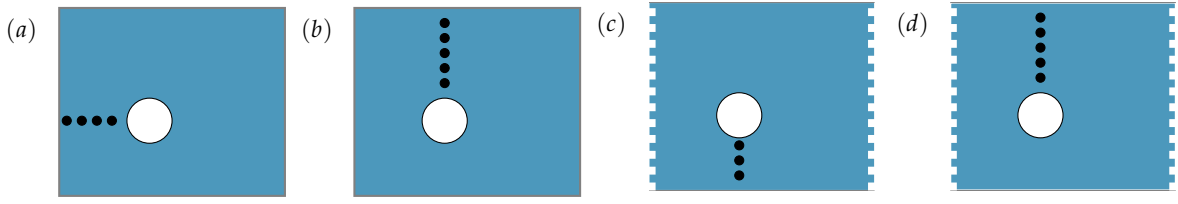
\begin{figure}[h]
\centering
  \begin{tikzpicture}
      \draw[gray,fill={rgb,255: red,76; green,152; blue,188},thick,inner sep=5pt] (1,2.5) rectangle (4,5);

      \foreach \t in {5,...,8} {
        \draw[fill=black] ({0.1+0.2*\t},3.5) circle (0.06cm);
      }
      \draw[fill=white] (2.2,3.5) circle (0.3cm);
      \node at (0.5,4.7) {$(a)$};
  \end{tikzpicture}
  \begin{tikzpicture}
    \draw[gray,fill={rgb,255: red,76; green,152; blue,188},thick,inner sep=5pt] (1,2.5) rectangle (4,5);
  
    \foreach \t in {0,...,4} {
      \draw[fill=black] (2.2,{4+0.2*\t}) circle (0.06cm);
    }
    \draw[fill=white] (2.2,3.5) circle (0.3cm);

    \node at (0.5,4.7) {$(b)$};
\end{tikzpicture}
\begin{tikzpicture}
  \draw[white,fill={rgb,255: red,76; green,152; blue,188},inner sep=5pt] (1,2.5) rectangle (4,5);
  \draw[gray] (1,2.5) edge[-] (4,2.5) (1,5) edge[-] (4,5);
  \draw[white,line width=5pt,dashed] (1,2.5) edge[-] (1,5) (4,2.5) edge[-] (4,5);
  
  \foreach \t in {0,...,4} {
    \draw[fill=black] (2.2,{2.7+0.2*\t}) circle (0.06cm);
    }
    \draw[fill=white] (2.2,3.5) circle (0.3cm);

  \node at (0.5,4.7) {$(c)$};
\end{tikzpicture}
\begin{tikzpicture}
  \draw[white,fill={rgb,255: red,76; green,152; blue,188},thick,inner sep=5pt] (1,2.5) rectangle (4,5);
  \draw[gray] (1,2.5) edge[-] (4,2.5) (1,5) edge[-] (4,5);
  \draw[white,line width=5pt,dashed] (1,2.5) edge[-] (1,5) (4,2.5) edge[-] (4,5);

  \foreach \t in {0,...,4} {
    \draw[fill=black] (2.2,{4+0.2*\t}) circle (0.06cm);
    }
    \draw[fill=white] (2.2,3.5) circle (0.3cm);

  \node at (0.5,4.7) {$(d)$};
\end{tikzpicture}
  \caption{Panels $(a)$ and $(b)$ feature a smooth puncture defined on a lattice with only smooth boundaries.
  The strings of $X$ operators defined only on VV qubits (as in panel $(a)$) or only on CC qubits (as in panel $(b)$) are equivalent.
  Panels $(c)$ and $(d)$ feature a lattice with smooth and rough boundaries, with a smooth punctured carved out from the inside.
  The logical $X$ shown running from the smooth puncture to the boundary in panels $(c)$ and $(d)$ are equivalent up to an embedded logical $X$.
  }
  \label{fig:smooth-everywhere}
\end{figure}
However, these two objects are equivalent up to stabilizer.
This equivalence is captured by $\Omega_X$.

Furthermore, consider a lattice with two smooth and rough boundaries as shown in panels $(c)$ and $(d)$ of fig. \ref{fig:smooth-everywhere}.
The two representations of the logical $X$ operator shown running from the smooth puncture to the boundary are equivalent up to an embedded logical $X$ operator, and $X$ stabilizers.
This equivalence is also captured by $\Omega_X$.

With these comments, we proceed to the proof of theorem \ref{thm:logicalX}.

\begin{proof}
  We start from first principles.
  Recall that we defined the matrix $\INT = \left(\1_B \otimes \1_S | \1_T \otimes \1_A\right)$. 
  The logicals are defined as
  \begin{align*}
    \ker{\h_Z + \h_Z'}/\rs{\h_X + \h_X' + \INT}~.
  \end{align*}

  \textbf{Part 1:} $\ker{\h_Z + \h_Z'}$

  We would like to understand the structure of vectors $\alpha \in \ker{\h_Z + \h_Z'}$.
  We shall assume that $\alpha$ is not supported on the interior of the puncture $B \times S \union T \times A$.
  If $\alpha$ is supported within the interior, it can be removed using the single-qubit operators described by $\INT$.
  As shown by lemma \ref{lem:embeddedX}, the embedded logical $X$ operators are unaffected by the puncture.

  We shall argue that $\alpha$ ought to have a certain structure using a proxy.
  Let us define $a \in \field_2^{C \times V}$ as
  \begin{align}
    \label{eq:g}
    a = \h_Z \alpha = \h_Z' \alpha~.
  \end{align}
  The only occasion when $a$ is non-trivial is when $\alpha$ is supported on the boundary.
  If not, $\alpha$ is either a stabilizer or logical belonging to the code that is unaffected by the puncture.
  We shall mod out by this set, and this will correspond to $\Omega_X$.

  The VV and CC portions of eq. \ref{eq:g} stipulate that
  \begin{enumerate}
    \item $a \subseteq \im{\h_T\1_{M\setminus B}} \cap \im{\h\1_{B^c}} \otimes \field_2^S$;
    \item $a \subseteq \field_2^T \otimes \im{\h_S^t\1_{N\setminus A}} \cap \im{\h\1_{A^c}}$
  \end{enumerate}
  respectively.

  Equivalently, this means that
  \begin{align}
    \label{eq:kerhz}
    \alpha \in (\ker{\h_{T^c}} \otimes \field_2^S | \bz_{m^2}) \union (\bz_{n^2}| \field_2^T \otimes \ker{\h_{S^c}^t})~.
  \end{align}

  \textbf{Part 2:} $\rs{\h_X + \h_X' + \INT}$

  Since the operators we are interested in only lie outside the puncture by assumption, we may ignore $\INT$ and need only concern ourselves with $\rs{\h_X + \h_X'}$.
  Let $g \in \ker{\h_{T^c}}$ and $x \in \rs{\h_A}$ such that $a\h_A = x$.
  Its product $g \otimes x$ is clearly in $\ker{\h_{T^c}} \otimes \rs{\h_A}$ and therefore in the kernel of $\h_Z + \h_Z'$.
  We shall show that this vector lies in the span of the $X$ stabilizers as well.
  Indeed, it can be expressed as
  \begin{align*}
    (g \otimes a)\left(\h_X + \h_X'\right) = g \otimes x~.
  \end{align*}
  In a similar manner, we can show that any vector in $\rs{\h_B^t} \otimes \ker{\h_{S^c}^t}$ is in the row span of $\h_X + \h_X'$.

  As was already shown in lemma \ref{lem:indpt}, any vector $\alpha \in \field_2^{n^2 + m^2}$ that is in the row span of
  \begin{enumerate}
    \item $(\ker{\h_{T^c}} \otimes \field_2^S / \rs{\h_A} | \bz_{m^2})$
    \item $(\bz_{n^2}| \field_2^T / \rs{\h_B^t} \otimes \ker{\h_{S^c}^t})$
  \end{enumerate}
  do not lie in the row span of the stabilizer $\h_X + \h_X'$.

  This completes the proof.
\end{proof}

\section{Wormholes}
\label{sec:wormholes}

We have now established how to construct punctures by carving out portions of the hypergraph product code.
On the surface code, punctures facilitate CNOT gates on encoded qubits via braiding, but it is limited.
This process maps physical (and logical) $X$ operators to $X$ operators and $Z$ operators to $Z$ operators.
In other words, it is a CSS-preserving operation.
To complete even just the Clifford group, we require operations that can map $X$ operators to $Z$ operators, on the physical and logical levels.
In particular, we need ways to perform logical single-qubit Clifford operations.
On a $2$-dimensional code, twist defects \cite{bombin2010topological,bombin2011clifford,yoder2017surface,brown2017poking} can be used to encode qubits, and also perform single-qubit Clifford gates on these qubits.

Twist defects however rely on symmetries of $2$-dimensional codes that do not naturally extend to general LDPC codes.
For instance, an error chain on the surface code has two frustrated stabilizers on either end regardless of the length of the chain.
LDPC codes however do not possess these properties.
For instance, expander codes have the property that the number of frustrated stabilizers grows with the size of the error.
For these reasons, we have to look for other ways of generalizing twist defects.

In a companion paper, we introduce a defect called a wormhole that addresses this issue.
Rather than rely on line-like defects, it builds upon and generalizes puncture defects.
Since we already know how to construct punctures on the hypergraph product code, it is natural to extend them to wormholes.

The key idea is to entangle stabilizers along the boundaries of punctures.
Doing so yields hybrid stabilizers whose weight does not scale with the size of the puncture.
As we shall see, these stabilizers are created by measuring two-qubit Pauli operators.
These measurements locally break the CSS nature of the code and serve as a resource to complete the Clifford group.

Let $\G = (V\union C,E)$ be a bipartite graph corresponding to a classical code $\C$.
Consider a hypergraph product of a graph $\G$ with itself.
As before, let $S \subseteq V$ and $T \subseteq C$ be connected subsets of variable nodes and check nodes respectively.
Furthermore the induced subgraphs do not overlap, i.e. they obey $N \cap T = M \cap S = \emptyset$.
These sets must be correctable, i.e., they obey conditions specified in definition \ref{def:correctable}.

The wormhole is created by entangling the stabilizers along the boundaries of two punctures.
This alters the structure of the code along the boundaries, and we must ensure that these enlarged regions remain correctable.
Hence, we need to strengthen the notion of correctability to include the neighborhoods that define the punctures.
\begin{definition}[Extended correctability]
\label{def:wormhole-correctable}
In addition to condition \ref{def:correctable}, wormholes will also need to obey
\[
  \ker{\h_{\Gamma(M)}} = \ker{\h_{\Gamma(N)}^t} = \emptyset~.
\]
\end{definition}
This will be necessary because the stabilizers on the boundary of the puncture will be removed to form hybrid stabilizers.
To argue that the logical operators that emerge have certain properties, we shall use the above extended correctability condition.
Equivalently these can be thought of as the conditions for a puncture defined using the sets $\tS := M$ and $\tT := N$.

With these constraints established, we can associate a smooth puncture to the product $T \times S$ and a rough puncture to the product $S \times T$.
These punctures will be used to construct a wormhole in the following sections.

\subsection{Measurements and hybrid-stabilizers}

Within the interior of the punctures, we perform the same measurements as we did to initialize a puncture -- single-qubit $X$ measurements within the smooth puncture and single-qubit $Z$ measurements within the rough puncture.
We then perform two-qubit measurements along the boundaries of the two punctures.
This yields hybrid stabilizers.

For what follows, it will be helpful to use the schematic for the smooth and rough puncture shown in fig. \ref{fig:wormholeSchematic}.
Recall that the boundaries of punctures are described as follows:
\begin{enumerate}
  \item Smooth puncture: $T \times \left(N \setminus A\right) \union \left(M \setminus B\right) \times S$.
  \item Rough puncture: $\left(N \setminus A\right) \times T \union S \times \left(M \setminus B\right)$.
\end{enumerate}
For any VV qubit $(u,u')$ or CC qubit $(c,c')$, we shall let $P(u,u')$ or $P(c,c')$ denote the single-qubit Pauli operator $P$ on that qubit.

\begin{figure}[h]
  \centering
  \begin{tikzpicture}
    \draw (0,0) rectangle (1.5,1.5); 
    \node at (0.75,0.75) {$T \times S$};
    \draw (1.9,1.9) rectangle (3.1,3.1); 
    \node at (2.5,2.5) {$B \times A$};
    \draw (1.9,3.25) rectangle (3.25,3.6); 
    \node at (2.55,3.4) {{$\scriptstyle (M \setminus B) \times A$}};
    \draw (3.25,1.9) rectangle (3.6,3.25); 
    \node[rotate=-90] at (3.4,2.6) {{$\scriptstyle B \times (N\setminus A)$}};
    \draw (1.9,0) rectangle (3.1,1.5); 
    \node at (2.5,0.75) {$T \times A$};
    \draw (0,1.9) rectangle (1.5,3.1); 
    \node at (0.75,2.5) {$B \times S$};
    \draw (-0.5,0) edge[-] (-0.5,1.5); 
    \node at (-0.75,0.75) {$T$};
    \draw (-0.5,1.9) edge[-] (-0.5,3.1); 
    \node at (-0.325,2.5) {$B$};
    \draw (-0.6,1.9) edge[-] (-0.6,3.6); 
    \node at (-0.8,2.5) {$M$};
    \draw (0,-0.5) edge[-] (1.5,-0.5); 
    \node at (0.75,-0.75) {$S$};
    \draw (1.9,-0.5) edge[-] (3.1,-0.5); 
    \node at (2.5,-0.3) {$A$};
    \draw (1.9,-0.6) edge[-] (3.6,-0.6); 
    \node at (2.5,-0.8) {$N$};
    \draw (0,3.25) rectangle (1.5,3.6); 
    \node at (0.75,3.4) {{$\scriptstyle (M \setminus B) \times S$}};
    \draw (3.25,0) rectangle (3.6,1.5); 
    \node[rotate=-90] at (3.4,0.75) {{$\scriptstyle T \times (N\setminus A)$}};
    
    \begin{scope}[yscale=1,xscale=-1,xshift=-250]
      \draw (0,0) rectangle (1.5,1.5); 
      \node at (0.75,0.75) {$S \times T$};
      \draw (1.9,1.9) rectangle (3.1,3.1); 
      \node at (2.5,2.5) {$A \times B$};
      \draw (1.9,3.25) rectangle (3.25,3.6); 
      \node at (2.55,3.4) {{$\scriptstyle (N \setminus A) \times B$}};
      \draw (3.25,1.9) rectangle (3.6,3.25); 
      \node[rotate=90] at (3.4,2.6) {{$\scriptstyle A \times (M \setminus B)$}};
      \draw (1.9,0) rectangle (3.1,1.5); 
      \node at (2.5,0.75) {$S \times B$};
      \draw (0,1.9) rectangle (1.5,3.1); 
      \node at (0.75,2.5) {$A \times T$};
      \draw (-0.5,0) edge[-] (-0.5,1.5); 
      \node at (-0.75,0.75) {$S$};
      \draw (-0.5,1.9) edge[-] (-0.5,3.1); 
      \node at (-0.325,2.5) {$A$};
      \draw (-0.6,1.9) edge[-] (-0.6,3.6); 
      \node at (-0.8,2.5) {$N$};
      \draw (0,-0.5) edge[-] (1.5,-0.5); 
      \node at (0.75,-0.75) {$T$};
      \draw (1.9,-0.5) edge[-] (3.1,-0.5); 
      \node at (2.5,-0.3) {$B$};
      \draw (1.9,-0.6) edge[-] (3.6,-0.6); 
      \node at (2.5,-0.8) {$M$};
      \draw (0,3.25) rectangle (1.5,3.6); 
      \node at (0.75,3.4) {{$\scriptstyle (N \setminus A) \times T$}};
      \draw (3.25,0) rectangle (3.6,1.5); 
      \node[rotate=90] at (3.4,0.75) {{$\scriptstyle S \times (M\setminus B)$}};
    \end{scope}
  \end{tikzpicture}
  \caption{Schematic to denote the wormhole.
  Smooth puncture on the left and a rough puncture on the right.}
  \label{fig:wormholeSchematic}
\end{figure}
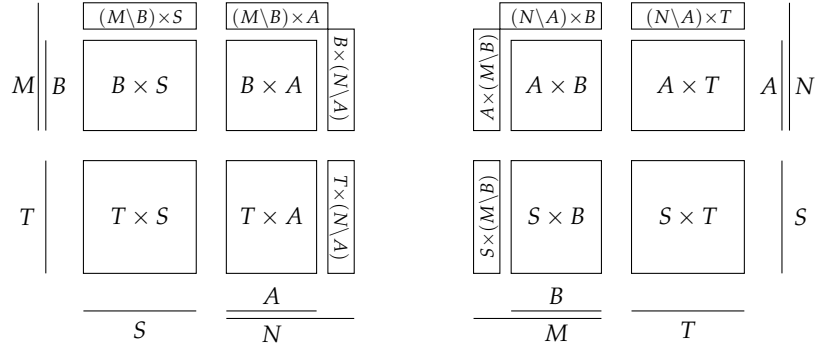

The hybrid-stabilizers are generated by the following measurements along these boundaries.
\begin{enumerate}
\item CC qubits: for every $c \in N\setminus A$, $c' \in T$, measure $X(c,c') \otimes Z(c',c)$; we denote this as
\begin{align*}
  (N \setminus A) \times T \leftrightarrow T \times (N \setminus A)~.
\end{align*}
\item VV qubits: for every $u \in S$, $u' \in M \setminus B$, measure $X(u,u') \otimes Z(u',u)$; we denote this as
\begin{align*}
  S \times (M\setminus B) \leftrightarrow (M \setminus B) \times S~.
\end{align*}
\end{enumerate}

These two-qubit measurements do not commute with the stabilizers located on the boundary of the puncture.
The next proposition will list the hybrid stabilizers we choose to resolve anti-commutations due to these measurements.
Suppose $P$ and $Q$ are some sets of variable and check nodes respectively, sets of the form $P \times Q$ will refer to $X$ stabilizer and $Q \times P$ to $Z$ stabilizers.
We shall write $P \times Q \leftrightarrow Q \times P$ to denote the hybrid stabilizer formed by pairing stabilizers in a natural way.
In other words, for $p \in P$ and $q \in Q$, we let $P \times Q \leftrightarrow Q \times P$ denote the hybrid stabilizers acting as $X$ on the support of $(p,q)$ and $Z$ on the support of $(q,p)$.
As a matter of convention, the operators with $X$ support are always denoted on the left, and those with $Z$ support on the right.

The proposition states that the set of stabilizers that we choose to form hybrids are those that are adjacent to the puncture minus those in the interior.
In other words, this is the set of stabilizers that live on the boundary of the punctures.
\begin{proposition}
  \label{prop:hybridStabs}
  Upon performing the measurements listed above, the new hybrid-stabilizers are associated with
  \begin{align*}
    (M \times N) \setminus (B \times A) \leftrightarrow (N \times M) \setminus (A \times B)~,
  \end{align*}
where the double-arrow denotes the one-to-one pairing between the two sets as described above.
\end{proposition}
\begin{proof}
  Since we begin with a puncture, certain stabilizers have already been removed from our code.
  These correspond to the stabilizers $B \times A$ of $X$ type and $A \times B$ of $Z$ type from the interior of respective punctures.
  We remove these stabilizers from a larger set corresponding to $M \times N$ and $N \times M$ respectively.
  With the interior carved out, this leaves only the boundary of the two punctures.

  Next, consider two nodes $c,v$ such that $c \in N$ and $v \in M$.
  The $Z$ stabilizer $(c,v)$ will anti-commute with the $X$ measurements if either $c \in N\setminus A$ or $v \in M \setminus B$.
  By symmetry, any $X$ measurement on the support of the $Z$ stabilizer $(c,v)$ will also act as $Z$ on the support of the $X$ stabilizer $(v,c)$.
  The two individual stabilizers are frustrated, but this can be resolved by pairing them.

  This produces the desired hybrid stabilizers.
\end{proof}

The weight of the hybrid-stabilizers is thus independent of the size of the code.
For this reason, this construction guarantees that we still have an LDPC code.

Thus when creating a wormhole, we begin as before by carving out certain portions of the code.
Then, we perform two-qubit measurements along the boundaries of these punctures.
We remove $X$ and $Z$ stabilizers from the code either because they lie within a puncture, or they anti-commute with a two-qubit measurement and are replaced by a hybrid stabilizer.

The code thus has the following $X$, $Z$ and hybrid (denoted $h$) stabilizers:
\begin{align*}
  \begin{matrix}
    X: & (V \times C) \setminus [(S \times T) \union (M \times N)]\\
    Z: & (C \times V) \setminus [(T \times S) \union (N \times M)]\\
    h: & (M \times N) \setminus (B\times A) \leftrightarrow (N \times M) \setminus (A \times B)\\
       & (N\setminus A) \times T \leftrightarrow T \times (N \setminus A)\\
       & S \times (M \setminus B) \leftrightarrow (M \setminus B) \times S~.
    \end{matrix}
\end{align*}
We conclude by reiterating the origin of these objects.
Note that there are two punctures that are used to create the wormhole, one smooth and one rough.
The set of all $X$ stabilizers corresponds to $V \times C$, but we subtract those stabilizers in these punctures.
This corresponds to $S \times T$ from the rough puncture and $M \times N$ from the smooth puncture and the hybrid stabilizers.
Similarly, the set of all $Z$ stabilizers corresponds corresponds to $C \times V$, but we subtract the stabilizers $T \times S$ from the smooth puncture and $N\times M$ from the rough puncture and the hybrid stabilizers.
The hybrid stabilizers are created using the boundaries of these sets and are stated in proposition \ref{prop:hybridStabs}.
The last two lines of hybrid stabilizers correspond to the two-qubit measurements used to generate the wormhole.
We remind the reader that in this notation, operators with $X$ support are to the left of the arrow, and those with $Z$ support are to the right of the arrow.

\subsection{Logical Pauli operators for wormholes}

When we create a wormhole from two punctures, there are two ways in which stabilizers are updated.
First, there are stabilizers within the puncture that are jettisoned because they anti-commute with the single-qubit measurements.
Second, there are stabilizers on the boundary of the puncture that are replaced by hybrid stabilizers.
This results in two sets of logical operators for the wormhole as we shall see below.
As in the case of punctures, there are two varieties of logical operators: loop-type operators and chain-type operators.
These logical operators are inherited from the underlying punctures.

The first set of logical operators can be described as the logical operators corresponding to the punctures $S \times T$ and $T \times S$.
Given the symmetry of the construction, there is a one-to-one correspondence between the loop-type logical $X$ operators around the rough puncture $S \times T$ and the loop-type logical $Z$ operators around the smooth puncture $T \times S$.

\begin{lemma}
\label{lem:wormhole-logicalZ}
  The logical $Z$ operators correspond to loop-type operators around one of the punctures.
  They come in two sets which are described as follows.
  \begin{enumerate}[label={}]
    \item \textbf{Type 1:} The first set of logical $Z$ operators correspond to the smooth puncture $T \times S$.
    \item \textbf{Type 2:} The second set of logical $Z$ operators corresponding to the smooth puncture $N \times M$.
  \end{enumerate}
\end{lemma}
\begin{proof}
  To show that these objects are no longer part of the stabilizer group but commute with the $X$ and $Z$ stabilizers, we point to the proof of theorem \ref{thm:logicalZ}.
  The development is similar save for the new stabilizers, the new two-qubit operators that were measured along the boundaries of the two punctures.

  \textbf{Type 1:} Note that the measurements surrounding the smooth puncture are of $Z$ type and will therefore commute with the loop-type logical $Z$ operators of type 1.
  Furthermore, these logical operators are defined only along the boundary of the puncture $T \times S$.
  A product of the two-qubit stabilizers is necessarily defined on both punctures, as $X$ on the puncture $T \times S$ and $Z$ on the puncture $S \times T$.
  Therefore this implies that the proposed logical operators of type 1 cannot be in the span of the stabilizers.

  \textbf{Type 2:} The logical operators of type 2 are defined on $(\Gamma(N) \setminus S) \times M \union N \times (\Gamma(M) \setminus T)$.
  Thus they do not interact with the two-qubit measurements.
  For the same reason, they cannot be expressed as a product of the two-qubit measurement operators.
  By definition \ref{def:wormhole-correctable}, these objects do not affect the embedded logical operators, and are themselves not in the span of the $Z$ stabilizer.
\end{proof}

Before proceeding to the conjugate logical operators, it will be useful to highlight a symmetry of this construction.
For logical $Z$ operators, we chose the vector of $Z$ operators around the puncture $T \times S$ for type 1 and $N \times M$ for type 2.
Equivalently we could have chosen the vector of $X$ operators around the puncture $S \times T$ for type 1 or $M \times N$ for type 2.
The next lemma states that these two choices are equivalent.

\begin{lemma}
  Every logical loop-type operator for a wormhole has two equivalent representations: an $X$ type loop around one puncture or a $Z$ type loop around the other.
\end{lemma}
\begin{proof}
  We shall deal with each type in turn.

  \textbf{Type 1:}

  Consider loop-type logical $X$ operators that emerge from the puncture $S \times T$.
  These logical operators are supported on $(N\setminus A) \times T \union S \times (M\setminus B)$.
  Each qubit on this boundary has a unique partner on the other boundary $T \times (N\setminus A) \union (M\setminus B) \times S$.
  By symmetry, there is a one-to-one correspondence between the loop-type logical $X$ operators on $S \times T$ and the loop-type logical $Z$ operators on $T \times S$.
  These can be mapped to one another because of the two-qubit measurements.

  \textbf{Type 2:}

  Let $\tG_S$, $\g_T$ be the matrices whose rows span $\ker{\h_S^t}$ and $\ker{\h_T}$ respectively.

  Let $g \in \tG_S^t$ and $f \in \g_T^t$ be any two rows of $\tG_S^t$ and $\g_T^t$ respectively.
  By the considerations above and theorem \ref{thm:logicalZ}, we can define $\alpha_Z$ as a loop-type operator around $T \times S$, where
  \begin{align*}
    \alpha_Z := (g \otimes f)(\h_N \otimes \1_M | \1_N \otimes \h_M^t)~.
  \end{align*}
  Similarly, $\alpha_X$ can be defined as a loop-type operator around $S \times T$, where
  \begin{align*}
    \alpha_X := (f \otimes g)(\1_M \otimes \h_N | \h_M^t \otimes \1_N)
  \end{align*}
  is also a logical operator.

  To show that these are in the span of the stabilizers, note that the hybrid stabilizers are given by
  \begin{align*}
    (\h_N \otimes \1_M | \1_N \otimes \h_M^t)_Z + (\h_A \otimes \1_B | \1_A \otimes \h_B^t)_Z \leftrightarrow (\1_M \otimes \h_N| \h_M^t \otimes \1_N)_X + (\1_B \otimes \h_A| \h_B^t \otimes \1_A)_X~.
  \end{align*}
  Thus the operator
  \begin{align*}
    (g \otimes f) (\h_N \otimes \1_M | \1_N \otimes \h_M^t)_Z \leftrightarrow (f\otimes g)(\1_M \otimes \h_N | \h_M^t \otimes \1_N)_X
  \end{align*}
  maps the loop-type logical $Z$ operator to the loop-type logical $X$ operator.

  Since this is true for arbitrary $f$ and $g$, any operator in the space can be mapped between one puncture and the other.
\end{proof}

The logical $X$ operator for the wormhole are products of chain-type operators.
The form of these operators are given in theorem \ref{thm:logicalX}.

\begin{lemma}
  For every loop-type logical $Z$ operator $L_Z$ around $T \times S$ and the unique loop-type logical $X$ operator $L_X$ on $S \times T$ corresponding to $L_Z$, let the conjugate chain-type operators be $Q_X$ and $Q_Z$ respectively.
  The product $Q_XQ_Z$ is the conjugate logical operator to the operator $L_Z$.
\end{lemma}
\begin{proof}
  Following the proof of theorem \ref{thm:logicalX}, the logical chain-type operators evidently commute with the $X$ and $Z$ stabilizers and are not spanned by them.

  From the symmetry of the construction, any overlap with the two-qubit measurement operators always occurs in pairs if at all.
  Therefore these operators commute with the two-qubit stabilizers.

  Furthermore since the two-qubit measurement operators are only supported on the boundary, it cannot span the logical chain-type operators.
  Since the chain-type operators anti-commute with the logical loop-type operators, it cannot be expressed as a product of stabilizers alone.
\end{proof}

\section{Code deformation}
\label{sec:code-def}

Having described how to create defects, we can now proceed to discuss how to use them.
Logical transformations will be effected using a technique called code deformation (see for instance \cite{bombin2011clifford,anderson2014fault,bombin2009quantum,vuillot2019code}).
The core idea behind this technique is to perform a sequence of $T-1$ elementary transformations of a code $\C =: \C^{(1)}$, obtaining codes $\C^{(2)},...,\C^{(T-1)}, \C^{(T)}$ in the process.
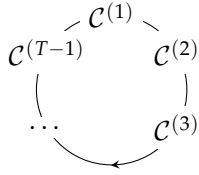
\begin{center}
\begin{tikzpicture}[decoration = {markings,mark = between positions 0.1 and 1 step 0.3333 with {\arrow[>=stealth]{<}}}]
      \draw[postaction = decorate] (0, 0) circle [radius = 1cm];
      \path (90 :1cm) node[fill=white]{$\mathcal{C}^{(1)}$}
            (30 :1cm) node[fill=white]{$\mathcal{C}^{(2)}$}
            (330:1cm) node[fill=white]{$\mathcal{C}^{(3)}$}
            (210:1cm) node[fill=white]{$\dots$}
            (150:1cm) node[fill=white]{$\mathcal{C}^{(T-1)}$};
    \end{tikzpicture}
\end{center}
Each elementary transformation is comprised of measurements of Pauli operators.
As shown in the schematic, the overall result is to leave the codespace globally unchanged, i.e. $\C^{(1)} = \C^{(T)} = \C$, but the logical operators of the code may, and hopefully will, undergo a non-trivial transformation.
Since this transformation maps all Pauli operators to Pauli operators, the resulting operation must be a logical Clifford operation.

We begin by reviewing code deformation to highlight some useful properties.
Rather than focus right away on hypergraph product codes, we step back and study code deformation as it applies to general quantum codes.
Our intent is to track the transformation of the logical operators.

\subsection{Non-mixing}
\label{subsec:nonmixing}
For all $t \in \{1,...,T\}$, the quantum error correcting code $\C^{(t)}$ is defined on $n$ qubits for some fixed $n$.
At step $t$, $\C^{(t)}$ is the eigenspace of the stabilizer group $\S^{(t)}$.

The logical operators of the code are denoted $\L^{(t)}$.
These are the objects that we wish to track as we transform the code.
Let $\Z := \S^{(t)} \cap \S^{(t+1)}$ be the operators that are common to  both $\S^{(t)}$ and $\S^{(t+1)}$.

The following lemma states that when we transition from $\C^{(t)}$ to $\C^{(t+1)}$, the logical operators that need to be updated are either removed entirely or mapped by multiplying by a stabilizer element.
\begin{lemma}
  \label{cor:mult}
  If a logical operator $L \in  \L^{(t)}$ anti-commutes with the measurement of $S^{(t+1)} \in \S^{(t+1)}/ \Z$, then it is either
  \begin{enumerate}
    \item $L$ is moved to the space of errors; or
    \item there exists a unique $S' \in \S^{(t)} \setminus \Z$ such that $L \mapsto LS'$.
  \end{enumerate}
\end{lemma}
\begin{proof}
  Consider any element $S \in \S^{(t+1)} \setminus \Z$.
  One of two things can happen to $S$ as we transition from step $t$ to $t+1$:
  \begin{enumerate}
    \item $S$ can commute with all of $\S^{(t)}\setminus \Z$, i.e. $S \in \L^{(t)}$.
    \item $S$ can commute with some element $S' \in \S^{(t)}\setminus \Z$.
    This element $S'$ can be chosen uniquely because if there were any other element $S_{0}$ that anti-commuted with $S$, we map $S_{0} \to S_{0}S'$.
  \end{enumerate}

  In turn, this leads to two possibilities for the logical operators $L^{(t)}$.
  Suppose we have an operator $L \in \L^{(t)}$ that anti-commutes with an element $S \in \S^{(t+1)}\setminus \Z$.
  We then update the logical as follows:
  \begin{enumerate}
    \item If $S$ is an operator in $\L^{(t)}$, then we remove the operator $L$ from the logical operators.
    It must now be an error.
    \item If there exists an element $S' \in \S^{(t)} \setminus \Z$ such that $S$ anti-commutes with $S'$, then $L \mapsto LS'$.
  \end{enumerate}
  This proves the claim.
\end{proof}

In addition, the sets $\S^{(t)}$ and $\L^{(t)}$ can also exchange operators.
The operators in $\S^{(t)} \setminus \S^{(t+1)}$ that commute with all of $\S^{(t+1)}$ will be transformed into logical operators.
For instance, this happens when we create a new logical qubit in the surface code by forming a puncture.
Recall that stabilizer generators and errors can be partitioned into pairs such that a stabilizer generator $S$ and error $E$ only anti-commute with each other, and commute with all other operators.
When a stabilizer $S$ is transformed into a logical operator, the conjugate errors $E$ becomes the unique conjugate error.
If there is no unique conjugate error, then this space cannot be used to store a qubit.
For instance this happens when we have a smooth puncture on a lattice with only rough boundaries.
The string of $X$ from the smooth boundary of the puncture cannot be terminated on the boundary.
Thus creating a smooth puncture on a lattice with only rough boundaries cannot be used to store a qubit (because there are redundant checks).

The operators in $\L^{(t)}$ that are in the span of $\S^{(t+1)} \setminus \S^{(t)}$ will be removed from the stabilizer group.

This analysis proves that logical operators transform linearly as summarized by the following lemma.
\begin{lemma}
  For $t \in \{1,...,T\}$, let $\L^{(t)}$ denote the set of logical operators of the code $\C^{(t)}$.
  Let $\mathbb{L}^{(t)} \in \field_2^{k \times n}$ and $\mathbb{S}^{(t)} \in \field_2^{(n-k)\times n}$ be the generator matrix for this space.
  There exists a matrix $Q^{(t)}$ such that we can write the logical operators $\mathbb{L}^{(t)}$ as
  \begin{align}
    \mathbb{L}^{(t+1)} = Q^{(t)}
    \begin{pmatrix}
      \mathbb{L}^{(t)}\\
      \mathbb{S}^{(t)}
    \end{pmatrix}
  \end{align}
\end{lemma}

As we proceed with code deformation, we will encounter problems unique to codes that carry several logical qubits.
We define below the notions of non-mixing and small transformations as guidelines for studying such transformations.

First, there may potentially be several ways of updating the logical operators.
This is because there is no preferred basis for us to express the logical operators in the intermediary steps.
Equivalently, the matrix $Q^{(t)}$ is not unique as there could be several different ways of expressing the logical operators over the course of code deformation.
However, the global transformation $Q = Q^{(T)}Q^{(T-1)}\ldots Q^{(2)}Q^{(1)}$ generated by the entire sequence of code deformation is unique if we choose the same logical operator basis for $Q^{(1)} = Q^{(T)}$.

For $t \in \{1,...,T\}$, let $\L^{(t)}$ denote the set of logical operators of the code $\C^{(t)}$.
Let
\begin{align*}
  \langle L^{(t)}_{j} \rangle_j := \langle L^{(t)}_g \rangle_g \times \langle L^{(t)}_b \rangle_{b}
\end{align*}
be a partition of the set of logical operators $\L^{(t)}$ into good and bad operators.
These sets are defined such that the operators in the set $g$ all have weight above some threshold, whereas those in $b$ have weight below this threshold.
Furthermore, assume that $\langle L^{(t)}_g \rangle$ contains $k' < k$ independent operators.
The set of qubits defined by $\{ L^{(t)}_b \}_b$ shall be considered as gauge qubits.
We wish to avoid the space of gauge qubits interacting with our logical qubits and to this end, define the non-mixing condition.

\begin{definition}[Non-mixing]
  \label{def:nonmixing}
  We say that code deformation is \emph{non-mixing} with respect to the partition if there exists a direct-sum decomposition
  \begin{align*}
    Q^{(t)} = Q^{(t)}_g \oplus Q^{(t)}_b~,
  \end{align*}
  where matrices $Q^{(t)}_g$, $Q^{(t)}_b$ only act on the spaces $\langle L^{(t)}_g \rangle_g$, $\langle L^{(t)}_b \rangle_b$ respectively.
  Each elementary step in code deformation is small with respect to $\Q^{(t)}_{g}$ if $\Q^{(t)}_g$ is rank $k'$ over the good partition.
\end{definition}

By guaranteeing that an operation is non-mixing with respect to this partition, we can show that the gauge qubits do not affect the logical qubits.
The constraint on the rank will guarantee that none of the good logical operators are mapped to either the stabilizers or gauge operators over the course of code deformation.

The non-mixing condition is important because we cannot guarantee that the number of logical operators will remain a constant over the course of code deformation.
In general, hypergraph product codes need not be translation invariant like the toric code; even if we maintain a puncture of a fixed radius, the number of logical operators created by this puncture could change as it moves.
If we move a puncture by enlarging it and then shrinking it, this could also change the number of logical operators supported by the puncture.
However the two conditions on the high-weight operators regulate their transformation.

To illustrate, we consider encoding logical qubits on the surface in a slightly unusual way.
Consider the pair of punctures on the surface code shown in fig. \ref{fig:nonmixing}.
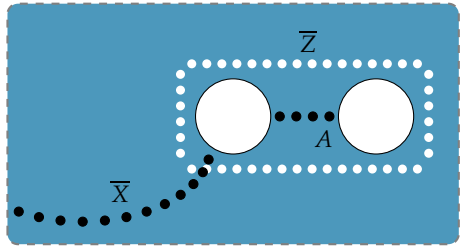
\begin{figure}[htp]
\centering
\begin{tikzpicture}
    \draw[gray,dashed,fill={rgb,255: red,76; green,152; blue,188},thick,inner sep=5pt,rounded corners] (-1,1.8) rectangle (5,5);
    
    \draw[fill=white] (2,3.5) circle (0.5cm);
    \draw[fill=white] (3.9,3.5) circle (0.5cm);

    \foreach \x in {1,...,16} {
      \fill[white] ({1.25 + .2*\x},4.2) circle (0.06cm);
      \fill[white] ({1.25 + .2*\x},2.8) circle (0.06cm);
    }
    \node at (3,4.5) {$\overline{Z}$};
    \foreach \y in {1,...,6} {
      \fill[white] (1.3,{2.8 + 0.21*\y}) circle (0.06cm);
      \fill[white] (4.6,{2.8 + 0.21*\y}) circle (0.06cm);
    }
    \foreach \x in {1,...,4} {
      \draw[fill=black] ({2.4 + 0.22*\x},3.5) circle (0.06cm);
    }
    \node at (3.2,3.2) {$A$};
    \foreach \t in {-3,...,8} {
      \draw[fill=black] ({1.7*sin(10*\t)},{3.1-cos(10*\t)}) circle (0.06cm);
    }
    \node at (0.5,2.5) {$\overline{X}$};
\end{tikzpicture}
\caption{A pair of punctures used to encode a single logical qubit on the surface code.}
\label{fig:nonmixing}
\end{figure}
This pair shall be treated as a single entity that encodes a logical qubit.
The logical $Z$ operator is the loop encircling the pair, denoted $\overline{Z}$.
The logical $X$ operator is a string of $X$s running from the boundary of a puncture to the boundary of the lattice, denoted $\overline{X}$.
The operator $A$ running between punctures is a low-weight string of $X$s and this logical operator is a potential liability.
We therefore treat it as a gauge qubit.
When encoding information, we only store logical information using the qubit defined by $\overline{Z}$ and $\overline{X}$.
So long as we do not braid using $A$ or drag another defect between these two punctures, this troublesome chain will not cause a problem.
In this way the operation will be non-mixing because the operator $A$ will never be entangled with other qubits of interest.

\subsection{Code deformation on the hypergraph product code}
\label{subsec:codedefhgp}

We now return to hypergraph product codes, and in this section shall discuss how to perform Clifford gates with the help of an ancilla.
We begin by recalling lemma $1$ from \cite{krishna2019topological}.
It stated that we can perform all Clifford gates on a qubit of interest, labelled $1$, with the help on an ancillary qubit, labelled $a$, as follows.

\begin{lemma}
  \label{lem:claim}
  Let $A$ and $B$ be distinct, non-trivial single-qubit Pauli operators.
  Let $S$ and $T$ be two Pauli operators, not necessarily distinct.
  The two-qubit measurements $A_1S_a$ and $B_1T_a$, together with all single-qubit Pauli measurements on qubit $a$ are sufficient to generate the single-qubit Clifford group on qubit $1$.
\end{lemma}

First, we point out that regardless of whether the qubit $1$ is an embedded logical qubit, or merely another wormhole, the ancilla qubit(s) shall be encoded in a wormhole.
Second, we will require these wormholes to encode $Y$ resource states.
For reasons that will become clear shortly, we will find it difficult to perform single-qubit $Y$ measurements on the logical level.
If however we are provided ancilla qubits that are prepared in the $Y$ state, we may use these objects catalytically to perform a $Y$ measurement.
This is necessary, at least as per lemma \ref{lem:claim}, to complete the Clifford group.

In addition to these single-qubit gates, we also require entangling gates between multiple logical qubits.
The following circuit shows how this too can be accomplished on the logical level with the help of an ancilla and Pauli measurements.
\begin{figure}[h]
\centering
  \begin{tikzpicture}
    \node at (-3.8,0.5) {$\ket{\psi}$};
    \node at (-3.8,1.5) {$\ket{\phi}$};
    \draw (-3.5,0.5) edge[-] (-2,0.5);
    \draw (-3.5,1.5) edge[-] (-2,1.5);
    \fill[black] (-2.75,0.5) circle (0.07cm);
    \fill[black] (-2.75,1.5) circle (0.07cm);
    \draw (-2.75,0.5) edge[-] (-2.75,1.5);

    \node at (-1.5,1) {$=$};

    \node at (-0.8,0.5) {$\ket{\psi}$};
    \node at (-0.8,1) {$\ket{0}$};
    \node at (-0.8,1.5) {$\ket{\phi}$};
    \draw (-0.5,0.5) edge[-] (0,0.5);
    \draw (-0.5,1) edge[-] (0,1);
    \draw (-0.5,1.5) edge[-] (0,1.5);
    \draw (0,0.8) rectangle (0.7,1.7);
    \node at (0.35,1.25) {{$ZX$}};
    \draw (0.7,1) edge[-] (1,1);
    \draw (0,0.5) edge[-] (1,0.5);
    \draw (1,0.3) rectangle (1.7,1.2);
    \node at (1.35,0.75) {{$ZZ$}};
    \draw (1.7,1) edge[-] (2,1);
    \draw (2,0.8) rectangle (2.4,1.2);
    \node at (2.2,1) {$X$};
    \draw (0.7,1.5) edge[-] (2.7,1.5);
    \draw (2.4,1) edge[-] (2.7,1);
    \draw (1.7,0.5) edge[-] (2.7,0.5);
  \end{tikzpicture}
  \caption{A circuit to perform controlled-$Z$.}
\end{figure}
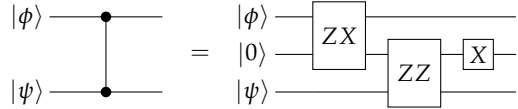
This completes the requirements to perform Clifford gates.
In the next subsection, we shall study how exactly to perform these measurements.

\subsection{Measurements and traceability}
The measurements of Pauli operators described above will have to be performed on the logical operator and the ancilla qubit encoded in a wormhole.
In addition, we use a regular puncture based qubit to perform the measurement.
This puncture shall be referred to as a needle.

We are interested in logical operators of the needle that can be measured fault tolerantly.
In turn these operators will be used to measure the logical Pauli operators of a wormhole or even an embedded logical qubit.
For instance, suppose we have a smooth puncture that has high weight $X$ and $Z$ operators.
The loop-type operator can be measured by shrinking the size of the puncture while simultaneously maintaining the size of the chain-type conjugate logical operator.
This would of course make the logical qubit susceptible to logical $Z$ error.
However this will not affect the measurement outcome so long as the logical $X$ operator remains high weight.
Similarly, the chain-type logical $X$ operator can be measured fault tolerantly by shrinking its size while maintaining the size of the loop-type $Z$ logical operator.
Unfortunately, the logical $Y$ operator of a puncture cannot be measured fault tolerantly as this would require that we minimize both the size of the chain-type operator as well as that of the loop-type operator.
The logical qubit would then become unprotected, and the measurement outcome error prone.
The impossibility to fault-tolerantly measure $Y$ operators will cause some problems as we shall see below.

Let $w$ denote a logical qubit, or sets of logical qubits whose state we wish to measure.
This could refer to a set of logicals on the wormhole, or an embedded logical, or some combination thereof.
Let $P_p$ refer to a logical operator of the puncture that is fault tolerantly measurable.
A logical operator $Q_w$ on system $w$ is said to be \emph{traceable} if there is a unitary operation $U$ implementable by code deformation such that
\begin{align*}
  U\conj P_p U = Q_w P_p~.
\end{align*}
Such operations will be used to measure traceable operators $Q_w$ in order to effect measurements of logical operators in lemma \ref{lem:claim}.
The operator $Q_w$ shall be measured using the standard ancilla-assisted way: we shall prepare the ancilla in an eigenstate of $P_p$, applying the unitary $U$ and then measure the operator $P_p$.

If $Q_w$ is either $X$ or $Z$, then we need to find an operator $P_p$ and an operation $U$ such that
\begin{align*}
  U\conj P_p U \to Q_w P_p~.
\end{align*}
On the other hand, if $Q_w = Y$, then the situation is complicated for reasons we will discuss shortly.
In such a case, we shall assume that we are given access to an another system $b$ prepared in a $Y$ state.
The aim is to perform the operation
\begin{align*}
  U\conj P_p U \to Y_w Y_b P_p~,
\end{align*}
which would not have been possible without the system $b$.
If we now use this operation to perform a $Y_wY_b$ measurement, then we can do so without affecting the state of system $b$.
In this way, the system $b$ serves as a catalyst and can be used for the next $Y$ measurement on $w$.
In the next subsection, we shall discuss how to obtain such a resource states to serve as catalysts.

We now make some remarks about traceability, but before doing so, we remark that this completes the requirements to perform Clifford gates on the system $w$.

We begin by noting that the advantage of the wormhole in this process is that it permits both the $X$ and $Z$ type logical operators associated to a defect to be traced.
This has been via example in our companion paper \cite{krishna2019topological}; whether or not a logical operator is traceable on a specific code is a code dependent question.
The other advantage of a wormhole is that it permits the use of a $Y$ resource state whose role is catalytic.
Without the wormhole, the $Y$ resource states would be consumed and we would therefore require a constant supply of these states.

In the case of the surface code, the new wormhole defects that we have introduced make it possible to trace both the $X$ and $Z$ type logicals associated to a wormhole \cite{krishna2019topological}.
Suppose that $P_p$ above is the logical $X$ operator of a smooth puncture.
As it traces the support of a logical $Y = iXZ$ operator, it will encounter the location when the $X$ and $Z$ logical operators cross.
As shown in fig. \ref{fig:nontraceable}, the trailing chain-type $X$ operator is mapped to the logical $Y$ operator and this has the effect of breaking the original protocol.
So instead of mapping $X_p$ to $Y_wX_p$ as desired, we are mapping it to $Y_wY_p$.
Completing the protocol would require measuring $Y_p$, but this cannot be done because it is not fault tolerantly measurable.

\begin{figure}[h]
  \begin{tikzpicture}
    \draw[gray,dashed,fill={rgb,255: red,76; green,152; blue,188},thick,inner sep=5pt,rounded corners] (0,1.8) rectangle (5,5);
    
    \draw[fill=white] (2.2,3.5) circle (0.5cm);

    \foreach \y in {1,...,15} {
      \fill[white] ({2.9 + 0.1*sin(60*\y)},{1.8 + 0.21*\y}) circle (0.06cm);
    }

    \foreach \t in {0,...,8} {
      \draw[fill=black] ({0.2*\t},{3.5+0.1*cos(70*\t)}) circle (0.06cm);
    }
    \node at (-0.5,4.7) {$(a)$};
  \end{tikzpicture}
  \def\p{20}
  \begin{tikzpicture}
    \draw[gray,dashed,fill={rgb,255: red,76; green,152; blue,188},thick,inner sep=5pt,rounded corners] (0,1.8) rectangle (5,5);
    
    \draw[fill=white] (4,3.5) circle (0.5cm);
    \foreach \t in {1,...,\p} {
      \fill[white] ({4 + 0.7*sin(\t*360/\p)}, {3.5 + 0.7*cos(\t*360/\p)}) circle (0.06cm);
    }

    \foreach \y in {1,...,15} {
      \fill[white] ({2.9 + 0.1*sin(60*\y)},{1.8 + 0.21*\y}) circle (0.06cm);
    }

    \foreach \t in {0,...,17} {
      \draw[fill=black] ({0.2*\t},{3.5+0.1*cos(70*\t)}) circle (0.06cm);
    }
    \node at (-0.5,4.7) {$(b)$};
  \end{tikzpicture}
  \caption{A puncture crossing its own path. The puncture first leaves a trail of $Z$s, and passes through a wormhole.
  Upon crossing $Z$, the logical $X$ is mapped to $XZ$.
  This operator is no longer guaranteed to be needle-measurable.}
  \label{fig:nontraceable}
\end{figure}
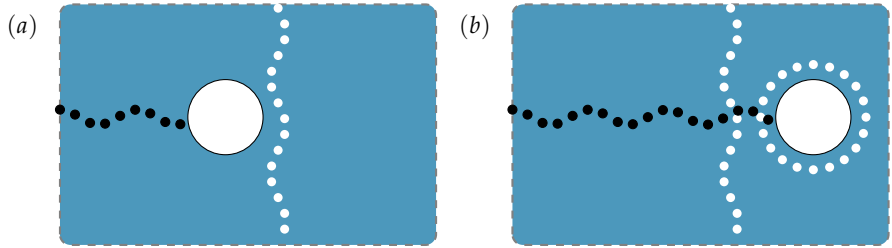

Not all is lost however; we may find that \emph{products} of logical $Y$ operators are traceable.
As an example, this was demonstrated in the case of the surface code with wormhole defects \cite{krishna2019topological}, so the framework we describe is sufficiently rich to enable the complete Clifford group in principle.

We would like to highlight that we are not providing a constructive approach to compiling specific logical operators.
Compiling the operation required to trace operators not only depends on the code, but also on the representation of the logical we are interested in performing.
The process described merely provides a framework within which to search for such an operation.
For instance, given a specific code, we could search for non-trivial Clifford operators that we can perform using brute force.
Once a set of Clifford operations has been found, these can be used as a basis to compose Clifford gates of interest using standard compiling tools.

We consider an example, perhaps perverse, to illustrate that simply having a wormhole and a puncture is insufficient to generate all Clifford gates.
Suppose we have two copies of the toric code that are disconnected from each other.
If we were to initialize a wormhole and a needle on one of the two codes, then clearly this is insufficient to perform all gates fault tolerantly on all logical qubits.
This is because there is clearly no way for the needle to move to the second code.
This suggests that in general, the ability to perform all gates may be closely tied to graph connectivity.
We will return to this idea in section \ref{subsec:pointlike}.

\subsection{Resource states}
\label{subsec:injection}

Clifford gates by themselves only generate a finite group \cite{nielsen2002quantum}.
Furthermore, the techniques that we have discussed above only map Pauli operators to Pauli operators.
Therefore they can at best generate logical Clifford operations.

As is well known, it is sufficient to add any gate that is not already in the Clifford group to achieve a universal gate set.
One such non-Clifford gate is the $T$ gate, defined as $T = e^{i\pi Z/8}$
We assume that we have access to physical $T$ gates and wish to construct a logical $T$ gate.
The technique presented here will mean that this logical $T$ gate is not inherently fault tolerant.
In turn the $T$ gate will be noisy and therefore need to be distilled \cite{bravyi2005universal}.


By its definition, the logical $T = e^{i\pi Z/8}$ gate can be executed on the support of the logical $Z$ operator.
 Suppose we have a logical qubit encoded in a smooth puncture prepared in the $\ket{+}$ state.
The logical $Z$ operator is a loop-type operator that is supported only on the boundary of the puncture.
We can inject the $T$ gate on to the code by performing an explicit circuit on this boundary.
One might hope that the stabilizers and logicals on the boundary obey some symmetry such as triorthogonality \cite{bravyi2012magic}.
In that case, the circuit to inject the $T$ gate could be made fault tolerant as we would merely require transversal physical $T$ gates in order to implement the transversal logical $T$ gate.
However we do not assume that the puncture obeys these symmetries, and thus the circuit to perform the $T$ is not fault tolerant.
Thus we would want to minimize the size of the circuit to minimize the number of faulty locations.
To this end, we may shrink the puncture i.e. reduce the size of the boundary that supports the logical qubit.
Upon performing the $T$ circuit, we increase the size of the puncture again to make it resistant to logical $Z$ type errors.
While doing so of course, the logical qubit is still subject to logical $Z$ errors and is thus subject to dephasing errors.
The end result is a noisy version of the $T$ state defined as $\ket{A} = T\ket{+}$.

Once we have several such logical qubits carrying potentially noisy $T$ states, we can perform state distillation on the hypergraph product code.
State distillation is a technique that uses several noisy $T$ states and produces fewer, but higher fidelity copies of the $T$ state.
These higher fidelity copies can then be used in the computation if they are sufficiently reliable.
State distillation only requires Clifford gates and Pauli measurements, and these are operations we already have the ingredients to perform.

To perform a logical $T$ on an embedded logical qubit, we can use the $T$ gate and a single-qubit teleportation circuit \cite{zhou2000methodology}.

In addition to using resource states to inject the $T$ gate, we also require resource states that are prepared in the $Y$ basis as was discussed in the previous section.
However since these resource states are used catalytically, they can be prepared once before the beginning of the computation.
To prepare these states, we follow a procedure similar to the preparation of a $T$ state.
We shall perform the circuit required to prepare a puncture qubit in a logical $Y$ state.
We note that the circuit to prepare the $Y$ state will have to be performed on the support of both the $X$ and $Z$ logical operators.
To minimize the size of this circuit, we therefore reduce the size of both the loop-type and the chain-type logical operators.
This circuit will likely not be fault tolerant; after performing the circuit we will have to increase the size of the loop-type and chain-type logical operators.
During this period the logical qubit supported on the puncture may be subject to depolarizing noise.

Once we have prepared several such logical qubits, we will have access to several noisy $Y$ resource states.
To purify theme, we will have to perform distillation.
Of course, there may be more optimal ways to prepare these states, but this is sufficient to generate the desired resource states.

\subsection{Point-like punctures}
\label{subsec:pointlike}

Before we conclude, we consider a scheme that involves the movement of point-like punctures.
This deviates slightly from the framework that we have discussed above, and since the punctures are point-like, the setup is not fault tolerant.
However, we feel that it still may help understand movement in these codes.

When discussing the surface code, traceability is relatively simple.
We can move a rough puncture around a smooth puncture and thereby perform a non-trivial Clifford operation.
Whether or not such an operation is possible is dictated by topology; a rough puncture can trace any closed loop of $Z$ operators.
The situation is not so simple in the case of hypergraph product codes.

In this subsection, we discuss when logical operators are traceable using a point-like puncture to serve as the needle.
Of course, using a point-like puncture is not fault tolerant as the loop-type logical operators are low-weight and therefore error prone.
This discussion will however shed light on when something non-trivial is possible.
It also illustrates that we can generalize braiding to graph-theoretic concepts.

We shall show that in the case of point-like punctures, traceability can be cast as walks on a graph. 
Thus whether or not a logical operator is traceable boils down to verifying whether a certain path on a graph exists.
This shows that it may be possible to efficiently verify when an operator is traceable.

Consider a point-like puncture, i.e. one created by removing a single stabilizer generator.
For the sake of illustration, this puncture is smooth.
Code deformation entails that there exist a series of steps such that the point-like puncture corresponds to $T_j \times S_j$ for $j = 1,...,N$.

Let $T_j = \{c_{j}\}, S_j = \{v_{j}\}$ be singleton sets and let $T_j \times S_j$ be the associated point-like puncture.
The logical $Z$ operator $\alpha$ that emerges is then just the support of the $Z$ stabilizer $(c_j,v_j)$ i.e. $\alpha_j := \Gamma(c_j) \times v_j \union c_j \times \Gamma(v_j)$.
Let the conjugate logical operator be $\beta_j$.
Let us study how these objects transform as we transition from step $j$ to step $j+1$.

\textbf{Moving a single step:} 
Suppose we consider moving this puncture by changing $T_{j} \to T_{j+1} = \{c_{j+1}\}$, where $c_{j}$ and $c_{j+1}$ share a bit $u_{j}$ in their common neighborhood as shown in fig. \ref{fig:pointlike}.

\begin{figure}[h]
\centering
  \begin{tikzpicture}
      \draw (0,0) node[csquare=c-0-1];
      \xdef\radius{0cm}
      \xdef\level{0}
      \xdef\nbnodes{1}
      \xdef\degree{(4+1)} 
      \foreach \ndegree/\form in {4/cdisc,4/square}{
        \pgfmathsetmacro\nlevel{int(\level+1)}
        \pgfmathsetmacro\nnbnodes{int(\nbnodes*(\degree-1))}
        \pgfmathsetmacro\nradius{\radius+0.8cm}
        \foreach \div in {1,...,\nnbnodes} {
          \pgfmathtruncatemacro\src{((\div+\degree-2)/(\degree-1))}
          \path (c-0-1) ++({\div*(360/\nnbnodes)-180/\nnbnodes}:\nradius pt) node[\form=c-\nlevel-\div];
          \draw (c-\level-\src) -- (c-\nlevel-\div);
        }
        \xdef\radius{\nradius}
        \xdef\level{\nlevel}
        \xdef\nbnodes{\nnbnodes}
        \xdef\degree{\ndegree}
      }

      \path (c-2-1) node[csquare=b];

      \node (c) at (0,-0.4) {$c$};
      \node (v) at (0.68,0.12) {$u$};
      \node (c') at (1.8,0) {$c'$};
      \node (c') at (1.3,0.8) {$c_{1}$};
      \node (c') at (0.74,1.3) {$c_{2}$};
     \begin{scope}[yscale=1,xshift=150]
        \draw (0,0) node[cdisc=c-0-1];
        \xdef\radius{0cm}
        \xdef\level{0}
        \xdef\nbnodes{1}
        \xdef\degree{(4+1)} 
        \foreach \ndegree/\form in {4/csquare,4/disc}{
          \pgfmathsetmacro\nlevel{int(\level+1)}
          \pgfmathsetmacro\nnbnodes{int(\nbnodes*(\degree-1))}
          \pgfmathsetmacro\nradius{\radius+0.8cm}
          \foreach \div in {1,...,\nnbnodes} {
            \pgfmathtruncatemacro\src{((\div+\degree-2)/(\degree-1))}
            \path (c-0-1) ++({\div*(360/\nnbnodes)-180/\nnbnodes}:\nradius pt) node[\form=c-\nlevel-\div];
            \draw (c-\level-\src) -- (c-\nlevel-\div);
          }
          \xdef\radius{\nradius}
          \xdef\level{\nlevel}
          \xdef\nbnodes{\nnbnodes}
          \xdef\degree{\ndegree}
        }
        \draw[] (0,-2) node {};
        \node (v) at (0,-0.4) {$v$};
    \end{scope}
    \end{tikzpicture}
  \caption{A point-like puncture centered at $c,v$.
    The check $c'$ is connected to $c$ via the variable node $u$.}
  \label{fig:pointlike}
\end{figure}
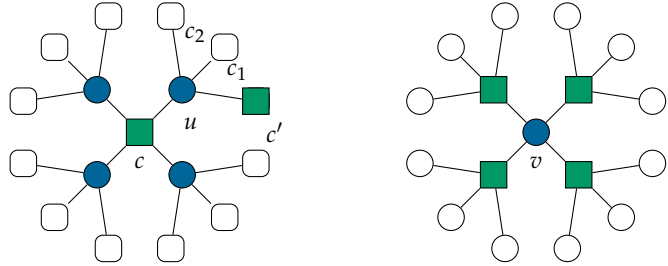

In the growth phase, we would measure $X$ on the qubit $(u,v_j)$.
This anti-commutes with all the $Z$ stabilizers that are incident to $(u,v_j)$, i.e. $(c_1,v_j)$, $(c_2,v_j)$ and $(c_{j+1},v_j)$.
The logical $Z$ operator $\alpha_j$ (corresponding to $(c_j,v_j)$) also anti-commutes with this operation.
These objects are updated per the stabilizer update rule.
We multiply the operators $(c_1,v_j)$, $(c_2,v_j)$ and $\alpha_j$ by $(c_{j+1},v_j)$.
The stabilizer $(c_{j+1},v_{j})$ is itself removed from the stabilizer group.

In the contraction phase, we measure the stabilizer $(c_{j},v_{j})$, returning it to the stabilizer group.
This anti-commutes with the measurement $X(u,v_{j})$.
This also anti-commutes with the logical operator $\beta_{j}$.
To resolve this anti-commutation relation, we map $\beta_j \to \beta_{j+1} := \beta_j X(u,v)$.
We then discard $X(u,v)$ from the stabilizer group.

Thus the logical operator $\beta_j$ has grown by a single qubit.
However, we have not yet returned to the code space.
The logical operator $\beta_{j+1}$ still anti-commutes with the operators $(c_1,v_j)$ and $(c_2,v_j)$.
These objects have still not been returned to the stabilizer group.
We highlight this matter because it is this issue that does not allow us to fit the movement of a point-like puncture into the framework described in the previous sections.

There are two concerns associated with these frustrated stabilizers:
\begin{enumerate}
  \item Will they return to the stabilizer group?
  \item Will the weight of these objects grow or remain upper bounded by some constant?
\end{enumerate}

First, since the path we traverse corresponds to a logical, it must commute with all the stabilizers.
Thus at some point during the course of the point-like puncture moving, it will commute with each stabilizer that it frustrated.
Exactly when these operators will be returned to the stabilizer group will depend on the path being traversed.

Secondly, the weight of these stabilizers is always guaranteed to be upper bounded by a constant.
In fact, these frustrated operators are always pairwise products of the puncture at step $j$ and themselves.
For instance, consider the example above where we transitioned by one step, from $j$ to $j+1$.
In the very next step, suppose the next point to be removed corresponds to the stabilizer $(c_{j+2},v_j)$.
The stabilizer $(c_{j+1},v_{j})$ is returned to the stabilizer group.
Therefore the frustrated stabilizers $(c_1,v_j)$ and $(c_2,v_j)$ are now multiplied by $(c_{j+2},v_{j})$.
This will continue until we measure another qubit in the support of $(c_1,v_j)$ and $(c_2,v_j)$, at which point they will return to the stabilizer group.

For a point-like puncture, this analysis shows that the logical can grow one qubit at a time.
With this insight, we can cast the problem of whether or not a logical is traceable as a graph problem.
In this problem, we first consider a logical operator, say $Q$, that has no $Y$ operators in its support.
We use this operator to define a graph $\G_{Q}$  as follows.
The stabilizers that are adjacent to the qubits become the vertices of $\G_{Q}$ and the qubits in the support of $Q$ become the edges of $\G_{Q}$.
If there exists a sequence of stabilizers that we can puncture, each connected by a single-qubit then this path becomes traceable.

In particular, this can be cast as Eulerian cycle \cite{moore2011nature}.
Eulerian cycle is an efficient algorithm that can be stated as follows:

\begin{algorithm}[H]
\caption{Eulerian cycle}
\label{alg:euler}
\begin{algorithmic}[1]
  \State \textbf{Input:} Graph $\G = (V,E)$.
  \State \textbf{Output:} A path on the graph such that every edge is traversed exactly once if it exists.
\end{algorithmic}
\end{algorithm}

It is well known that an Eulerian cycle exists only when the degree of each vertex in the graph is even.
Thus given a graph with $n$ vertices, the existence of an Eulerian cycle can be verified efficiently.
This example shows that braiding may generalize to a purely graph-theoretic concept. 
Moreover, there may exist an efficient algorithm to answer when a logical is traceable. 

\section{Conclusions}
We have provided a general framework to implement Clifford gates on hypergraph product codes.
This framework is based on code deformation, and generalizes defect based encoding from topological codes.
In particular, we generalize wormhole defects introduced in a companion paper \cite{krishna2019topological}.
These defects ensure that the code remains LDPC at each step of code deformation.

In contrast to a previous scheme suggested by Gottesman, these operations are defined on a single block.
The generalized punctures that we obtain are capable of encoding several logical qubits.
We discussed a framework that is rich enough to permit all Clifford gates on encoded qubits.
Whether a particular code permits these gates is a code dependent question.
Finally, we discussed the movement of point-like charges on these graphs.
These defects serve to illustrate that something non-trivial can be accomplished on hypergraph product codes.
Furthermore they demonstrate how braiding can be generalized to a purely graph-theoretic notion.

Of course, this is merely a proof of concept, and there is a lot of work to be done in the future.
In no particular order, we discuss some issues that need to be addressed; this is by no means a complete list.
Using point-like punctures helps us understand that traceability in certain instances is connected to Eulerian cycle.
As punctures become larger however, it is unclear how this algorithm will generalize.
Furthermore, we would like efficient algorithms that can verify whether a given representation of a logical is traceable.

We believe that addressing these questions will be intimately connected to the specific code we wish to use.
It is of course important to consider specific classes of codes to understand which Clifford gates we can generate using this process.
At this juncture however, this seems premature as there is as yet no consensus as to which hypergraph product codes offer the best performance.
As the theory progresses, this will likely be informed by decoding algorithms, but perhaps the ability to perform Clifford gates could also factor into this choice.
The punctures may exhibit symmetries that do not require state distillation to prepare resource states.
Since these codes are no longer local, it is unclear what sorts of gates can be implemented transversally, and which cannot.

Addressing these questions will help establish the role of hypergraph product codes in quantum computation.

\section{Acknowledgements}
AK would like to thank Christophe Vuillot, Daniel Gottesman, Vivien Londe and Nicolas Delfosse for discussions.
AK also acknowledges the support of the Fonds de Recherche - Nature et Technologie (FRQNT) for the B2X scholarship.
This work was partially funded by Canada's NRC and NSERC. David Poulin is a CIFAR Fellow in the Quantum Information Science program.

\bibliographystyle{unsrtabbrev}
\bibliography{references}

\begin{thebibliography}{10}

\bibitem{aharonov1997fault}
D.~Aharonov and M.~Ben-Or.
\newblock Fault-tolerant quantum computation with constant error.
\newblock In {\em Proceedings of the twenty-ninth annual {ACM} symposium on
  Theory of computing}, pages 176--188. ACM, 1997.

\bibitem{aliferis2006quantum}
P.~Aliferis, D.~Gottesman, and J.~Preskill.
\newblock Quantum accuracy threshold for concatenated distance-3 codes.
\newblock {\em Quantum Information \& Computation}, 6(2):97--165, 2006.

\bibitem{kitaev1997quantum}
A.~Y. Kitaev.
\newblock Quantum computations: algorithms and error correction.
\newblock {\em Russian Mathematical Surveys}, 52(6):1191--1249, 1997.

\bibitem{knill1998resilient}
E.~Knill, R.~Laflamme, and W.~H. Zurek.
\newblock Resilient quantum computation: error models and thresholds.
\newblock In {\em Proceedings of the Royal Society of London A: Mathematical,
  Physical and Engineering Sciences}, volume 454, pages 365--384. The Royal
  Society, 1998.

\bibitem{kitaev2003fault}
A.~Y. Kitaev.
\newblock Fault-tolerant quantum computation by anyons.
\newblock {\em Annals of Physics}, 303(1):2--30, 2003.

\bibitem{bravyi1998quantum}
S.~B. Bravyi and A.~Y. Kitaev.
\newblock Quantum codes on a lattice with boundary.
\newblock {\em arXiv preprint quant-ph/9811052}, 1998.

\bibitem{bombin2006topological}
H.~Bombin and M.~A. Martin-Delgado.
\newblock Topological quantum distillation.
\newblock {\em Physical Review Letters}, 97(18):180501, 2006.

\bibitem{kovalev2013fault}
A.~A. Kovalev and L.~P. Pryadko.
\newblock Fault tolerance of quantum low-density parity check codes with
  sublinear distance scaling.
\newblock {\em Physical Review A}, 87(2):020304, 2013.

\bibitem{gottesman2014fault}
D.~Gottesman.
\newblock Fault-tolerant quantum computation with constant overhead.
\newblock {\em Quantum Information \& Computation}, 14(15-16):1338--1372, 2014.

\bibitem{bravyi2009no}
S.~Bravyi and B.~Terhal.
\newblock A no-go theorem for a two-dimensional self-correcting quantum memory
  based on stabilizer codes.
\newblock {\em New Journal of Physics}, 11(4):043029, 2009.

\bibitem{bravyi2010tradeoffs}
S.~Bravyi, D.~Poulin, and B.~Terhal.
\newblock Tradeoffs for reliable quantum information storage in 2{D} systems.
\newblock {\em Physical {R}eview {L}etters}, 104(5):050503, 2010.

\bibitem{nickerson2013topological}
N.~H. Nickerson, Y.~Li, and S.~C. Benjamin.
\newblock Topological quantum computing with a very noisy network and local
  error rates approaching one percent.
\newblock {\em Nature communications}, 4:1756, 2013.

\bibitem{campagne2018deterministic}
P.~Campagne-Ibarcq, E.~Zalys-Geller, A.~Narla, S.~Shankar, P.~Reinhold,
  L.~Burkhart, C.~Axline, W.~Pfaff, L.~Frunzio, R.~Schoelkopf, et~al.
\newblock Deterministic remote entanglement of superconducting circuits through
  microwave two-photon transitions.
\newblock {\em Physical Review Letters}, 120(20):200501, 2018.

\bibitem{axline2018demand}
C.~J. Axline, L.~D. Burkhart, W.~Pfaff, M.~Zhang, K.~Chou, P.~Campagne-Ibarcq,
  P.~Reinhold, L.~Frunzio, S.~Girvin, L.~Jiang, et~al.
\newblock On-demand quantum state transfer and entanglement between remote
  microwave cavity memories.
\newblock {\em Nature Physics}, page~1, 2018.

\bibitem{kurpiers2018deterministic}
P.~Kurpiers, P.~Magnard, T.~Walter, B.~Royer, M.~Pechal, J.~Heinsoo,
  Y.~Salath{\'e}, A.~Akin, S.~Storz, J.-C. Besse, et~al.
\newblock Deterministic quantum state transfer and remote entanglement using
  microwave photons.
\newblock {\em Nature}, 558(7709):264, 2018.

\bibitem{tillich2014quantum}
J.-P. Tillich and G.~Z{\'e}mor.
\newblock Quantum {LDPC} codes with positive rate and minimum distance
  proportional to the square root of the blocklength.
\newblock {\em IEEE Transactions on Information Theory}, 60(2):1193--1202,
  2014.

\bibitem{poulin2008iterative}
D.~Poulin and Y.~Chung.
\newblock On the iterative decoding of sparse quantum codes.
\newblock {\em Quantum Information and Computation}, 8(10), 2008.

\bibitem{leverrier2015quantum}
A.~Leverrier, J.-P. Tillich, and G.~Z{\'e}mor.
\newblock Quantum expander codes.
\newblock In {\em Foundations of Computer Science (FOCS), 2015 IEEE 56th Annual
  Symposium on}, pages 810--824. IEEE, 2015.

\bibitem{fawzi2018efficient}
O.~Fawzi, A.~Grospellier, and A.~Leverrier.
\newblock Efficient decoding of random errors for quantum expander codes.
\newblock In {\em Proceedings of the 50th Annual ACM SIGACT Symposium on Theory
  of Computing}, pages 521--534. ACM, 2018.

\bibitem{fawzi2018constant}
O.~Fawzi, A.~Grospellier, and A.~Leverrier.
\newblock Constant overhead quantum fault-tolerance with quantum expander
  codes.
\newblock {\em arXiv preprint arXiv:1808.03821}, 2018.

\bibitem{kovalev2018numerical}
A.~A. Kovalev, S.~Prabhakar, I.~Dumer, and L.~P. Pryadko.
\newblock Numerical and analytical bounds on threshold error rates for
  hypergraph-product codes.
\newblock {\em Physical Review A}, 97(6):062320, 2018.

\bibitem{grospellier2018numerical}
A.~Grospellier and A.~Krishna.
\newblock Numerical study of hypergraph product codes.
\newblock {\em arXiv preprint arXiv:1810.03681}, 2018.

\bibitem{liu2018neural}
Y.-H. Liu and D.~Poulin.
\newblock Neural belief-propagation decoders for quantum error-correcting
  codes.
\newblock {\em Physical {R}eview {L}etters}, 122(20):200501, 2019.

\bibitem{panteleev2019degenerate}
P.~Panteleev and G.~Kalachev.
\newblock Degenerate quantum {LDPC} codes with good finite length performance.
\newblock {\em arXiv preprint arXiv:1904.02703}, 2019.

\bibitem{richardson2008modern}
T.~Richardson and R.~Urbanke.
\newblock {\em Modern coding theory}.
\newblock Cambridge university press, 2008.

\bibitem{calderbank1996good}
A.~R. Calderbank and P.~W. Shor.
\newblock Good quantum error-correcting codes exist.
\newblock {\em Physical Review A}, 54(2):1098, 1996.

\bibitem{steane1996multiple}
A.~Steane.
\newblock Multiple-particle interference and quantum error correction.
\newblock {\em Proceedings of the Royal Society A}, 452(1954):2551--2577, 1996.

\bibitem{nielsen2002quantum}
M.~A. Nielsen and I.~Chuang.
\newblock Quantum computation and quantum information, 2002.

\bibitem{vuillot2019code}
C.~Vuillot, L.~Lao, B.~Criger, C.~G. Almudever, K.~Bertels, and B.~Terhal.
\newblock Code deformation and lattice surgery are gauge fixing.
\newblock {\em New Journal of Physics}, 2019.

\bibitem{bombin2010topological}
H.~Bombin.
\newblock Topological order with a twist: Ising anyons from an abelian model.
\newblock {\em Physical Review Letters}, 105(3):030403, 2010.

\bibitem{bombin2011clifford}
H.~Bombin.
\newblock Clifford gates by code deformation.
\newblock {\em New Journal of Physics}, 13(4):043005, 2011.

\bibitem{yoder2017surface}
T.~J. Yoder and I.~H. Kim.
\newblock The surface code with a twist.
\newblock {\em Quantum}, 1:2, 2017.

\bibitem{brown2017poking}
B.~J. Brown, K.~Laubscher, M.~S. Kesselring, and J.~R. Wootton.
\newblock Poking holes and cutting corners to achieve {C}lifford gates with the
  surface code.
\newblock {\em Physical Review X}, 7(2):021029, 2017.

\bibitem{anderson2014fault}
J.~T. Anderson, G.~Duclos-Cianci, and D.~Poulin.
\newblock Fault-tolerant conversion between the {S}teane and {R}eed-{M}uller
  quantum codes.
\newblock {\em Physical {R}eview {L}etters}, 113(8):080501, 2014.

\bibitem{bombin2009quantum}
H.~Bombin and M.~A. Martin-Delgado.
\newblock Quantum measurements and gates by code deformation.
\newblock {\em Journal of Physics A: Mathematical and Theoretical},
  42(9):095302, 2009.

\bibitem{krishna2019topological}
A.~Krishna and D.~Poulin.
\newblock Topological wormholes.
\newblock {\em arXiv}, 2019.

\bibitem{bravyi2005universal}
S.~Bravyi and A.~Kitaev.
\newblock Universal quantum computation with ideal {C}lifford gates and noisy
  ancillas.
\newblock {\em Physical Review A}, 71(2):022316, 2005.

\bibitem{bravyi2012magic}
S.~Bravyi and J.~Haah.
\newblock Magic-state distillation with low overhead.
\newblock {\em Physical Review A}, 86(5):052329, 2012.

\bibitem{zhou2000methodology}
X.~Zhou, D.~W. Leung, and I.~L. Chuang.
\newblock Methodology for quantum logic gate construction.
\newblock {\em Physical Review A}, 62(5):052316, 2000.

\bibitem{moore2011nature}
C.~Moore and S.~Mertens.
\newblock {\em The nature of computation}.
\newblock OUP Oxford, 2011.

\end{thebibliography}

\end{document}